\documentclass[ecta,nameyear, draft]{econsocart}
 \pdfoutput=1

\usepackage[polish, english]{babel}
\usepackage{natbib}
\usepackage{amsmath}
\usepackage{multirow}
\usepackage{array}
\usepackage{bbm}
\usepackage{amsthm}
\usepackage{amsfonts}
\usepackage{setspace}
\usepackage{booktabs}
\usepackage{caption}
\usepackage{amssymb}
\usepackage{ragged2e}
\usepackage{graphicx}
\usepackage{enumerate}
\usepackage{subcaption}
\usepackage{mathrsfs}
\usepackage{microtype}
\usepackage{xr}

\usepackage{graphicx}
\usepackage[retainorgcmds]{IEEEtrantools}
\usepackage{pgf}
\RequirePackage[colorlinks,citecolor=blue,linkcolor=blue,urlcolor=blue,pagebackref,  breaklinks=true]{hyperref}
\usepackage{enumitem}
\usepackage{ifthen}
\usepackage{etoolbox}
\startlocaldefs

\theoremstyle{plain}

\newtheorem{ass}{Assumption}

\newtheorem{prop}{Proposition}
\newtheorem{cor}{Corollary}
\newtheorem{theo}{Theorem}

\theoremstyle{remark}

\newtheorem{rem}{Remark}
\newtheorem{defi}{Definition}

\DeclareMathOperator*{\Forall}{\Large \forall}
\DeclareMathOperator*{\argmin}{arg\,min}
\DeclareMathOperator*{\argmax}{arg\,max}

\newcounter{deferred}
\newcommand{\deferred}[2][]{%
  \ifstrempty{#1}
    {\stepcounter{deferred}\expandafter\gdef\csname temp\arabic{deferred}\endcsname{#2}}
    {\expandafter\gdef\csname temp#1\endcsname{#2}}%
}
\newcommand{\shownow}[1]{\csname temp#1\endcsname}
\newcommand{\deriv}[1]{\frac{\partial}{\partial #1}}
\newcommand{\derivsec}[1]{\frac{\partial^2}{\partial #1^2}}
\newcommand{\tderiv}[1]{\frac{\mathrm{d}}{\mathrm{d} #1}}
\endlocaldefs

\begin{document}
\begin{frontmatter}

\title{The Pond Dilemma with Heterogeneous Relative Concerns}
\runtitle{The Pond Dilemma}
\date{today}
\begin{aug}

\author[id=au1,addressref={add1}]{\fnms{Pawe\l{}}~\snm{Gola}\ead[label=e1]{pawel.gola@ed.ac.uk}}

\address[id=add1]{%
\orgdiv{School of Economics},
\orgname{University of Edinburgh}}
\end{aug}

\support{The paper was first circulated in September 2024, and this version is from October 2024. I would like to thank Miguel Ballester, Stuart Breslin, Antonio Cabrales, Ed Hopkins, Tai-Wei Hu, Emir Kamenica, Alastair Langtry, Andriy Zapechelnyuk, Yuejun Zhao, and the seminar audience in Edinburgh for helpful discussions and suggestions.  I gratefully acknowledge the support of the UKRI Horizon Europe Guarantee (Grant Ref: EP/Y028074/1). All errors are mine.}
\begin{abstract}
This paper explores team formation when workers differ in skills and their desire to out-earn co-workers. I cast this question as a two-dimensional assignment problem with imperfectly transferable utility and show that equilibrium sorting optimally trades off output maximisation with the need to match high-skill workers to co-workers with weak relative concerns. This can lead to positive (negative) assortative matching in skill even with submodular (supermodular) production functions. Under supermodular production, this heterogeneity in preferences benefits all workers and reduces wage inequality. With submodular production, the distributional consequences are ambiguous, and some workers become worse off. The model reveals that skill-biased technological change (SBTC) incentivises domestic outsourcing, as firms seek to avoid detrimental social comparisons between high- and low-skill workers, thus providing a compelling explanation for the long-term increase in outsourcing. Finally, the benefits of SBTC can trickle down to low-skill workers—but only those whose relative concerns are weak.
\end{abstract}


\end{frontmatter}

\section{Introduction}

Forming productive and durable teams requires more than finding workers with compatible skills: Preference and personality compatibility matters as well. Of particular importance are the relative concerns of the team-members. Many a famous sports teams and music bands have disintegrated because multiple members felt they deserved to be the biggest fish in that particular pond. Less anecdotally, there is strong empirical evidence that humans care about their relative position within the reference group and are willing to accept lower absolute wages to improve their relative earnings.\footnote{See, for example, \cite{Luttmer2005, Card2012, Perez-Truglia2020, Bottan2020}.} Moreover, the strength of these relative concerns--and their close cousin, competitiveness---differs across individuals, affecting their career and location choices.\footnote{See \cite{Buser2014} and \cite{Bottan2020}.}

The impact that relative concerns have on team formation should be of general interest to economists. Most importantly, these concerns distort sorting, and through that economy-wide output and wages. For example, a match between a low-skill worker who cares greatly about status and a high-skill worker who cares only about their own wage may well be output-maximising. Yet, it is not mutually beneficial, because the worker with strong relative concerns would be very unhappy about their low status in such a match. Even when a specific match between a high-low skill match remains viable, the low-skill worker still needs compensation for their low relative standing, which clearly affects wage inequality. Relative concerns may even influence firm boundaries: If---as argued by \cite{Nickerson2008}---social comparisons are more salient within than across firm boundaries, firms may salvage some of such output-maximising matches by outsourcing low-skill workers, thus avoiding potentially detrimental social comparison altogether. The size of these potential distortions depends on the difference in productivity between high- and low-skill workers, as this difference determines how much lower the low-skill worker's wage and status are. As a result, the presence of relative concerns creates a novel channel through which skill-biased technological change affects sorting, inequality and outsourcing.

In this paper, I study the impact that relative concerns have on the economy by casting the problem as a one-sided sorting problem with teams of two. Each team produces output, which depends on the heterogeneous skills of the two co-workers and can be freely divided among them. Workers' utility is a weighted sum of their wage (i.e, the amount of output the worker receives) and the difference between their own and their co-worker's wage. The weight with which this wage difference enters utility differs across workers, and represents the strength of relative concerns.
The heterogeneity of relative concerns renders utility  imperfectly transferable: If a worker has weaker relative concerns than their co-worker, then any increase in the worker's own wage increases their utility by less than it decreases the co-worker's utility. Fortunately, for my choice of utility function, the model admits a transferable utility representation, that is, there exists a re-scaling of each individual's utility which is perfectly transferable. The sum of the worker's and their co-worker's re-scaled utilities will be called the \emph{surplus} of the match; by standard results, equilibrium sorting maximises aggregate surplus of the economy.

This baseline model yields three main results. First, I fully characterise equilibrium sorting, payoffs and wages for a broad class of production functions and distributions of traits.
My characterisation relies on two additional assumptions. Assumption \ref{ass: common} stipulates, effectively, that the surplus function is supermodular (only) in the worker's skill and some index of the co-worker's skill and relative concerns, which I call the \emph{skill-preference index}. Assumption \ref{ass: copula} requires that the resulting copula of skill and the skill-preference index is symmetric. Assumption \ref{ass: common} implies that the social planner would like to match workers with high skill to co-workers with high values of the skill-preference index. Assumption \ref{ass: copula} ensures that such positive and assortative matching between skill and the index is feasible, and thus prevails in equilibrium.
I show that Assumptions \ref{ass: common} and \ref{ass: copula} are satisfied for large classes of specifications, specifically if either (a) the production function is  additive and the copula of traits is exchangeable, (b) skills are binary or (c) the production function is multiplicative and the joint distribution of traits is log-elliptical.

The equilibrium sorting optimally trades-off two motives against each other. The first motive is output maximisation: By standard arguments, if high-skill workers are complements (substitutes) with each other, then output is maximised by matching workers positively (negatively) and assortatively in the skills dimension. The second, and novel, motive is the maximisation of the welfare stemming from social comparisons: As high-skill workers earn higher wages, they impose a cost on their co-workers. This cost is minimised by matching them to co-workers with weak relative concerns. In the case of additive production, in particular, aggregate output does not depend on the sorting pattern, and thus the first motive is absent. Sorting between the workers' skill and their co-worker's strength of relative concerns is then negative and assortative, as conjectured by \cite{Frank1984}. More generally, the strength of equilibrium sorting in the skill dimension depends not only on the properties of the production function as in a standard sorting model, but also on the joint distribution of traits. As a result, it is perfectly possible for high-skill workers to be complements (substitutes) with each other, and yet for equilibrium sorting to remain perfectly negative (positive) and assortative in skills.

Second, the presence of heterogeneous relative concerns affects wage inequality in comparison to the standard model. because any high-skill workers matched to a less skilled co-worker must compensate that co-worker for their lower status by paying them more than their self-match wage (i.e., the wage they would receive by matching their identical clone). Of course, the amount of this compensation depends on the differences in productivity between high- and low-skill workers---the larger the difference in productivity, the stronger the difference in wages absent any compensation, and the larger the compensation needs to be. This indicates that, in some cases, improvements in high-skill workers productivity \emph{trickle down} to those low-skill workers who are matched with them! Furthermore, it also indicates that wage inequality in the presence of relative concerns is lower than it would be if all workers self-matched. As all workers self-match in the benchmark if high-skill workers are complements with each other, it follows that the presence of relative concerns decreases wage inequality in that case.

If high-skill workers are substitutes with each other, however, then sorting is negative and assortative and workers earn more than the self-match wage in the benchmark; indeed, it is perfectly possible that the benchmark wage distribution is less unequal than the self-match wage distribution. As self-matching always obtains if high-skill workers have sufficiently weaker relative concerns than low-skill workers, this implies that the presence of relative concerns can increase wage inequality.  In other words, while the presence of relative concerns always decreases within-firm wage inequality, it may increase between-firm wage inequality to such an extent that overall inequality raises. Interestingly, the conditions under which this can happen---weak (strong) relative concerns for high- (low-) skill workers---make the binary skill case isomorphic to a model in which all agents have inequity aversion, suggesting that aversion to within-firm inequity can increase overall inequality.

Third, I derive the impact that relative concerns have on agents' welfare. To avoid interpersonal welfare comparisons, I assume that each worker's utility depends on their personal strength of relative concerns and the intensity of interactions with their co-workers; the standard model corresponds then to the case where workers have relative concerns but never interact with co-workers. If high-skill workers are complements with each other, then all workers benefit from social comparisons: this is simply because they can always obtain the benchmark payoff by self-matching, and as many of them chose not to, they must be better off. Again, however, the substitutes case yields dramatically different conclusion. In that case, workers receive a payoff higher than the self-match payoff in the benchmark;
and yet, low-skill workers with strong relative concerns and high-skill workers with weak relative concerns choose to self-match, in order to avoid detrimental social comparisons. Therefore, social comparisons are welfare-reducing for these self-matching workers.

I build on that last insight to develop a theory of the boundaries of the firm. \citep{Nickerson2008} provide convincing case studies which demonstrate that outsourcing alleviates the salience of social comparisons. Thus, if outsourcing is possible but costly, teams that are identical in terms of skills may choose to draw the boundaries of their firms differently, simply because of differences in the strength of relative concerns: High-skill workers with strong relative concerns will keep their low-skill co-workers in-house, whereas high-skill workers with weak relative concerns will outsource them. Strikingly, this simple theory of the firm  implies that  skill-biased technological change (SBTC), by increasing within-firm wage inequality, makes outsourcing more beneficial for teams in which the high-skill worker has weaker relative concerns than the low-skill workers. Furthermore, while this has no impact on team composition, it does increase the number of firms that consist of workers of the same type, and thus increases sorting as measured in the data.

The rest of the paper is structured as follows. Section \ref{sec: rellit} discusses the related literature. Section~\ref{sec: model} develops the model. Section \ref{sec: char} characterises the equilibrium. Section \ref{sec: ToF} develops a social-comparisons-based theory of the firm, and uses it to explore the consequences of skill-biased technological changes for wage inequality, sorting and outsourcing.
 Section \ref{sec: cr} concludes.  Appendix \ref{app: proofs} contains omitted proofs and derivations.
\section{Related Literature}\label{sec: rellit}
\paragraph*{Sorting with Relative Concerns} I am aware of three articles and one book that study the problem of how workers sort into teams/firms if they have heterogeneous relative concerns. The seminal work by Robert Frank \citep{Frank1984, Frank1985} posed this important problem and unearthed the fundamental insights that the presence of heterogeneous relative concerns means that within-firm wage inequality is lower than productivity inequality, and that workers with stronger relative concerns end up having less skilled co-workers. \cite{Fershtman2006} extended the problem posed by Frank to include effort provision, finding that firms consisting of workers with strong and weak relative concerns require workers with strong relative concerns to exert more effort.\footnote{There is a large literature that studies the impact of relative concerns on effort provision \citep[see, for example, ][]{Hopkins2004, Hopkins2009}; this literature, however, assumes homogeneous relative concerns and is not concerned with sorting.} \cite{Langtry2023} differs from the other work on this topic (including mine), in that wages are set exogenously in his model, thus precluding high-skill workers from compensating low-skill workers for they lower status.\footnote{It is worth noting that the bulk of \cite{Langtry2023} is concerned with the altogether different problem of optimal choice of consumption on a network, when agents care about their neighbours' consumption.}
 For that reason, relative concerns affect wage inequality through sorting only, and the trickle-down effect---critical for my results---is absent.\footnote{In that sense,  \cite{Langtry2023} is actually closer  literature focusing on workers' who have homogeneous relative concerns and choose between two occupations \citep[e.g., ][]{Fershtman1993, Mani2004, Gola2024} than it is to \cite{Frank1984, Frank1985, Fershtman2006} and the present paper. In both \cite{Langtry2023} and the occupational choice literature, the firms/occupations do not internalise the externalities caused by relative concerns, and thus all of the impact of relative concerns happens through sorting, rather than wage setting. For that reason, relative concerns affect sorting even when all workers care about status equally.}

I make a number of important contributions to this small literature. Among others, this is the first paper to (a) provide analytical expressions for sorting and wages in settings with rich skill and preference heterogeneity, (b) allow a worker's productivity to depend on their co-worker's skill, (c) study how changes to the production function affect sorting, inequality and outsourcing in the presence of heterogeneous relative concerns, and (d) consider the impact that heterogeneous relative concerns have firms' boundaries.

\cite{Cabrales2008, Cabrales2008b} consider sorting in the presence of inequity aversion rather than relative concerns, and show that inequity aversion leads to more positive and assortative sorting in skill; a result that is also implied by my Proposition \ref{prop: sorting}.\footnote{\cite{Cabrales2008} motivate their focus on inequity aversion arguing that status concerns would produce counterfactual sorting in skills in the economy. While this is true when status concerns are homogeneous, one of the insights from the current work is that when status concerns are heterogeneous, then any degree of sorting in skills can be rationalised, irrespective of the properties of the production function.} In contrast to my work, they only allow for additively separable production functions, and so the surprising result that inequality aversion may increase overall wage inequality does not occur in their settings.

\paragraph*{Multidimensional Sorting}
This paper contributes to the literature on multidimensional, frictionless sorting with (a) transferable utility (representation)
\citep{Tinbergen1956, Galichon2016, Lindenlaub2017, Gola2021}.\footnote{Utility is imperfectly transferable in my model, but there exists a transferable utility representation \citep[see ][for a general discussion of sorting under imperfectly transferable utility]{Legros2007}. This implies, similarly to \cite{Clark2024}, that the equilibrium maximises appropriately defined total surplus in the economy, but does not---in general--- maximise total output.} With the exception of \cite{Gola2021}, this literature relies on a bi-linear surplus function and a Gaussian distribution of traits to provide closed form solutions. This article is the first one to (a) allow for one of the dimensions of heterogeneity to be a social preference rather than skill; (b) provide close form solutions for trait distributions that are not Gaussian and (c) consider a one-sided sorting model. I do this by leveraging a unique property of one-sided matching: It is always isomorphic to a two-sided model with a symmetric surplus function and identical distributions of traits on each side. This property is extremely useful, because under reasonable conditions on the surplus function \citep[see, for example, Proposition 11(b) in][]{Lindenlaub2017},  equilibrium sorting in such two-sided problems involves positive and assortative matching within each of the skill/preference dimensions.
In other words, one-sided sorting problems are more amenable to the introduction of multidimensional traits than two-sided problems are, because the assumption of identical traits distribution is very easy to satisfy.\footnote{The solution method in \cite{Tinbergen1956, Galichon2016, Lindenlaub2017} also leverages this property: Because Gaussian distributions are closed under linear transformations and surplus is bi-linear, there exist transformations of the workers' traits that have the same distributions on both sides of the market.\looseness=-1}

\paragraph*{Theory-of-the-Firm} My article follows a very old tradition in economics that contributes the need for concentrating economic activity within firms to transaction costs \citep{Coase1937, Williamson1971, Klein1978}. However, as noted in \cite{Coase1937}, the transaction costs theory-of-the-firm is not sufficient to explain why economic activity does not take place in one gigantic firm. The theory presented in this paper provides a clear answer to that question: Firm size is limited by the need to weaken social comparisons between high-skill workers who have weak relative concerns and low-skill workers with strong relative concerns.\footnote{The `property right' \citep{Grossman1986, Hart1990} and `incentive systems' \citep{Holmstrom1994, Holmstrom1999} provide complementary explanations for where firm boundaries are drawn.} In that, my theory formalises and expands upon \cite{Nickerson2008}, who propose an informal theory of the firm based on the need for a firm to manage the cost of  social comparison. In addition to being the first to formalise the `social comparison' theory of firm, I expand  on \cite{Nickerson2008} by considering agents who differ in the strength of relative concerns, which allows me to explain why seemingly identical firms make different outsourcing decisions.\footnote{Interestingly, the seeds of the social comparison based theory of the firm are present already in \cite{Coase1937}, who dismisses that theory's importance on the grounds that it would imply that entrepreneurs earn less than their employees. This implication, however, is incorrect as soon as one allows for heterogeneity in skills: In my model, `entrepreneurs'  are paid more than `employees', simply because they are more skilled. This is true even though  `entrepreneurs' are indeed taking a pay cut compared to the case where they self-match.} I also show that the `social comparisons' theory of the firm provides a  natural explanation for the rise in outsourcing in the recent decades: Simply put, by increasing the difference in productivity between high- and low-skill workers, skill-biased technological change has drastically increased the cost of within-firm social comparisons, thus increasing firms' incentives to outsource.

\paragraph*{Technology, Inequality and Outsourcing}
The pronounced increase in wage inequality since the 1970s is commonly attributed to skill-biased technological change \citep{Bound1992, Katz1992, Juhn1993}. The more recently documented disproportionately large increase in between-wage inequality over the same period \citep{Song2018, Tomaskovic2020} has spurned a growing theoretical literature that tries to explain its causes. The explanations include growing complementarities in production \citep{Freund2022}, trade liberalisation \citep{Trottner2022}, skill-biased technological change \citep{Cortes2023, Gola2023}, and falling span-of-control costs \citep{Gola2023}. In particular, \cite{Boerma2021} develops a model in which---as in the present paper---submodular production function plays an important role in explaining changes in within-firm inequality. In \cite{Boerma2021} within-firm inequality is present only when production is submodular, as this implies (partially) negative and assortative matching. In my model, in contrast, relative concerns' heterogeneity produces within-firm inequality even with supermodular production; submodular production is thus needed only to explain the rise in domestic outsourcing, an issue on which \cite{Boerma2021} is silent about.\looseness=-1

To the best of my knowledge, this is the first paper to provide a formal model linking (skill-biased) technological change with outsourcing. \cite{Bergeaud2024} show empirically that firms connected to broadband internet engage in domestic outsourcing more than firms without such connection.  \cite{Bergeaud2024} use the informal theory developed in \cite{Abraham1996} to explain why technological change may cause outsourcing. One of the reasons for outsourcing put forward by \cite{Abraham1996} is that, in the absence of outsourcing, within-firm wage inequality may be constrained by workers' inequity aversion.  I model inequity aversion/relative concerns explicitly, and highlight that skill-biased technological change---by creating pressure for higher wage differentials between high- and low-skill workers---naturally leads to more outsourcing, which further amplifies SBTC's impact on wage inequality.\looseness=-1

\section{Model}\label{sec: model}

There is continuum of workers, who differ along two dimensions---skill $x_1 \in I_{x_1} \subset \mathbb{R}$ and the strength of their relative concerns $x_2\in I_{x_2} \subset \mathbb{R}$. The joint distribution of $(x_1, x_2) = \mathbf{x}$ is denoted by $H: I_{\mathbf{x}} \to [0, 1]$, where $I_{\mathbf{x}}\equiv I_{x_1} \times I_{x_2}$ and $H$ has full support. The marginal distribution of $x_k$ will be denoted by $H_{x_k}$; for notational simplicity I will assume that $I_{x_1}, I_{x_2}$ are both closed sets and $\text{Pr}(X_2 \leq x_2 | x_1)$ is absolutely continuous in $x_2$ for any $x_1 \in I_{x_1}$.

Workers sort into teams of size two, which makes this a one-sided, one-to-one sorting problem.
A match between a worker with skill $x_1^k$ and a worker with skill $x_1^j$ produces output according to a symmetric, increasing and twice-continuously differentiable function $F: I_{x_1}^2 \to \mathbb{R}$. The output $F(x_1^k, x_1^j)$ is then endogenously split into the wages of the two workers.

 In contrast to standard assignment models, a worker's utility depends not only on their own wage, but also on their co-worker's wage. For simplicity, I will focus on the most natural utility function that captures these \emph{relative concerns}:
 \begin{equation}\label{eq: utildefinition}
U(w^k, w^j; x_2^k) \equiv w^k+x_2^k(w^k-w^j),
 \end{equation}
 where $w^k$ denotes the worker's own wage and $w^j$ denotes their co-worker's wage. Hence, in addition to their own wage, workers care also about the difference between their wage and their co-worker's wage, and the strength of these relative concerns is precisely $x_2$. To ensure that each worker's utility increases in their own wage, I assume that $-0.5 <\underline{x_2}\equiv \min I_{x_2}$.

Each agent's outside option is strictly lower than $F(x_1^k, x_1^k)/2$, the utility the agent would receive in a `self-match' (i.e., a match with a worker of the same type). Finally, note that if $x_2=0$ for all workers, then the model reduces to a standard \cite{Sattinger1979} sorting model. I will refer to this case the \emph{benchmark} and denote it by the subscript $B$.

\subsection{Imperfectly Transferable Utility}\label{sec: ITU}
If workers differ in their relative concerns, then utility becomes imperfectly transferable. To see this, consider any matched pair in which one worker cares less about relative wages than their co-worker ($x_2^k<x_2^j$). Clearly then, any increase in that worker's own wage increases their utility by less than it reduces the utility of their co-worker, and utility becomes imperfectly transferable.

Given that, and adapting from \cite{Legros2007}, in order to define the equilibrium we need to first specify the \emph{utility possibility frontier} $\psi: I_{\mathbf{x}}^2\times \mathbb{R} \to \mathbb{R}$, such that
\begin{IEEEeqnarray*}{rCl}
\psi(\mathbf{x}^j, \mathbf{x}^k, u)&\equiv& \max_{w^j} U\left(F(x_1^k, x_1^j)-w^j, w^j, x_2^k\right)\\
&& \text{subject to } U(w^j, F(x_1^k, x_1^j)-w^j, x_2^j) \geq u.
\end{IEEEeqnarray*}
In other words, the utility possibility frontier $\psi(\mathbf{x}^j, \mathbf{x}^k, u)$ is equal to the highest utility worker $\mathbf{x}^k$ can achieve  with a co-worker $\mathbf{x}^j$ if the co-worker receives utility of at least $u$. It is easy to see that in our specific case, we have that
\[\psi(\mathbf{x}^j, \mathbf{x}^k, u)=0.5(1+2x_2^k)\left(F(x_1^k, x_1^j)\left((1+2x_2^k)^{-1}+(1+2x_2^j)^{-1}\right)-2u(1+2x_2^j)^{-1} \right).\]
\subsection{The Competitive Equilibrium}
A function $\mu: I_{\mathbf{x}} \to I_{\mathbf{x}}$ is a \emph{sorting} if it  satisfies $\mu(\mu(\mathbf{x}))=\mathbf{x}$: That is, if the co-worker of $\mathbf{x}$'s co-worker is $\mathbf{x}$ themselves. A pure sorting $\mu$ is \emph{feasible} if $\mu(\mathbf{X}) \sim H$, that is, if the distribution of traits implied by the sorting is the same as the actual distribution of traits; the set of all feasible sorting is denoted by $S(H)$.

All worker's take the payoff function $u(\mathbf{x}^j)$ as given. A  sorting $\mu$ is \emph{individually rational} given a payoff function $u$ if
\begin{equation*} \label{eq: IR}\mu(\mathbf{x}^k)=\mathbf{x}^j \Rightarrow \mathbf{x}^j \in \argmax_{\mathbf{x}} \psi(\mathbf{x}^k, \mathbf{x}, u(\mathbf{x})). \end{equation*}

Finally, in equilibrium it must be the case that
\begin{equation}
    \label{eq: max} u(\mathbf{x}^k)=\max_{\mathbf{x}^j} \psi(\mathbf{x}^j, \mathbf{x}^k, u(\mathbf{x}^j)).
\end{equation}

\begin{defi}\label{defi: eqm}
A competitive equilibrium consists of a sorting $\mu^*: I_{\mathbf{x}}^2 \to [0, 1]$ and a payoff function $u^*: I_{\mathbf{x}} \to \mathbb{R}$, such that $\mu^*$ is feasible and individually rational given $u^*$, and $u^*$ satisfies  \eqref{eq: max}.
\end{defi}

\subsection{Transferable Utility Representation}\label{sec: TUR}

As discussed in Section \ref{sec: ITU}, the presence of heterogeneous relative concerns makes utility imperfectly transferable.
Fortunately, however, for our specific choice of utility function, the model admits a \emph{transferable utility representation}, which simplifies the analysis considerably.

\begin{defi}[\cite{Legros2007}]
A utility frontier $\psi$ is \emph{TU-representable} if there exists a function $\Pi: I_{\mathbf{x}} \to \mathbb{R}$, and a strictly increasing function $Z(\cdot; \mathbf{x})$  such that
\[\text{for all } \mathbf{x}^j, \mathbf{x}^k, \tilde{u} \qquad Z(\psi(\mathbf{x}^j, \mathbf{x}^k, Z^{-1}(\tilde{u}; \mathbf{x}^j)); \mathbf{x}^k)=\Pi(\mathbf{x}^j, \mathbf{x}^k)-\tilde{u}.\]
\end{defi}

In other words, a sorting model has a TU-representation if there exists an increasing transformation of each worker's utility function, which can depend on that workers' own type, that makes utility transferable within each matched pair. If such a transformation exists then, of course, one can simply work with the transformed utility functions.

To see that the present model is TU-representable, re-scale the utility function by defining \begin{equation}\label{eq: utilrescale}\tilde{u}(w^k, w^j; x_2^k)\equiv 2u(w^k, w^j; x_2^k)(1+2x_2^k)^{-1} %
=w^k-w^j+(w^k+w^j)(1+2x_2^k)^{-1}.  \end{equation}
Clearly, the re-scaled utility $\bar{u}$ is perfectly transferable within each match, because
\[2\psi(\mathbf{x}^j, \mathbf{x}^k, 0.5(1+2x_2^j)\tilde{u}))(1+2x_2^k)^{-1}=\Pi(\mathbf{x}^j, \mathbf{x}^k)-\tilde{u},\]
where
\begin{equation}\label{eq: Pi}
    \Pi(\mathbf{x}^j, \mathbf{x}^k)\equiv F(x_1^k, x_1^j)\left((1+2x_2^k)^{-1}+(1+2x_2^j)^{-1}\right).
\end{equation}

Intuitively, the initial specification of utility is imperfectly transferable, because worker $\mathbf{x}^k$ gains $1+2x_2^k$ utils from an extra dollar. Hence, if the worker has stronger relative concerns than their co-worker ($x_2^k>x_1^j$), then she gains more utility from an extra dollar of output than the co-worker loses. Critically, however, this difference in utility gained vs. lost depends only on the teammates preferences, and not on the output produced or its split. For that reason, one can achieve a transferable utility representation by dividing each worker's utility by $1+2x_2^k$, that is, the number of utils they gain from an extra dollar of wage.

\subsection{Assumptions}\label{sec: mtdsorting}
In order to characterise the equilibrium I will need two further assumptions. To state the first one, let me define the functions
\[L(\mathbf{x}^k; x_1', x_1)\equiv (F(x_1^k, x_1')-F(x_1^k, x_1))(1+2x_2^k)^{-1}, \quad V(l; x_1', x_1)=\text{Pr}(L(\mathbf{x}; x_1', x_1)\leq l). \]
\begin{ass}[Common Rankings]\label{ass: common}
For any $x_1''', x_1'', x_1', x_1 \in I_{x_1}$ such that $x_1'''>x_1''$ and $x_1'>x_1$, and any $\mathbf{x}^k\in I_{\mathbf{x}}$ we have that
\[V(L(\mathbf{x}^k; x_1''', x_1''); x_1''', x_1'')=V(L(\mathbf{x}^k; x_1', x_1); x_1', x_1)\equiv v_1(\mathbf{x}^k).\]
I will refer to $v_1(\mathbf{x}^k)$ as the \emph{skill-preference index}.
\end{ass}
As we will see more explicitly in the proof of Theorem \ref{theo: exuniq}, Assumption \ref{ass: common} implies that there exists an index of skill and preference, such that the surplus function is supermodular only in $(x_1^k, v_1(\mathbf{x}^j))$ and $(x_1^j, v_1(\mathbf{x}^k))$. In other words, Assumption \ref{ass: common}, implies that the strength of relative concerns will affect sorting only through the skill-preference index.\footnote{Assumption \ref{ass: common} is inspired by Assumption 1 in \cite{Gola2021}.}

To state the second assumption, I will require the concept of a bi-variate copula.

\begin{defi}[Copula]
A \emph{bivariate copula} is a supermodular function $C: [0, 1]^2 \to [0, 1]$ such that $C(0, v)=C(u, 0)=0$, $C(1, v)=v$ and $C(u, 1)$.

A function $C_{\mathbf{Y}}: [0, 1]^2 \to [0, 1] $ is the copula of the bivariate random vector $\mathbf{Y}$ if
\[C(\text{Pr}(Y_1 \leq y_1), \text{Pr}(Y_2 \leq y_2))=\text{Pr}(Y_1 \leq y_1, Y_2 \leq y_2).\]
\end{defi}
 Sklar's Theorem \citep{Sklar1959} ensures there exists a copula for every random vector, and that this copula is unique if the random vector is continuously distributed. For discrete random vectors, a continuum of copulas exists.
\begin{ass}[Properties of the Copula]\label{ass: copula}
The distribution $H$ and the production function $F$ are such that there exists an \emph{exchangeable} copula $C$ of $(X_1, v_1)$, that is, such a copula that $C(u, v)=C(v, u)$ for all $u, v \in [0, 1]$.
\end{ass}

Assumptions  \ref{ass: common} and \ref{ass: copula} are undoubtedly strong; and yet, they are satisfied for many natural combinations of assumptions about the production function $F$ and the traits distribution $H$. The following assumption describes three large classes of specifications for which Assumptions  \ref{ass: common} and \ref{ass: copula} are both satisfied.

\begin{ass}[Sufficient Conditions]\label{ass: suff}
The distribution $H$ and the production function $F$ satisfy at least one of the following conditions:
\begin{enumerate}[label={A3.\arabic*}]
 \item \emph{Additive Production with Exchangeable Distribution:} $F(x_1^k, x_1^j)=K(x_1^k)+K(x_1^j)$ and the copula $C_{i}$ of $(x_1, 1/(1+2 x_2))$ is exchangeable; or \label{item: additive}
 \item \emph{Binary Skills:} $x_1 \in \{l, h\}$, with $h>l$ and $\text{Pr}(X_1=h)=0.5$;
 \footnote{\label{foot: pneqhalf} If $p\equiv\text{Pr}(X_1=h) \neq 0.5$ then the economy consists of two completely separate sub-economies: one in which workers always self-match, as there are two few workers of the other skill for them to possibly match, and one consisting of an equal measure of low- and high-skill workers who choose whether to self- or cross-match. The first economy is trivial, and the second is isomorphic to an economy with $p=0.5$. Thus, the assumption that $p=0.5$ is without loss for the characterisation of equilibrium.}
 or\label{item: binary}
 \item \emph{Multiplicative Production with Log-Elliptical Distribution} $F(x_1^k, x_1^j)=A+((x_1^kx_1^j)^c-1)/c$, where $c \neq 0$ and $(\ln x_1, \ln (1+2x_2)) \sim EC_2(\Delta, \Omega; \phi)$, that is, the characteristic function of the joint distribution of $(\ln x_1, \ln (1+2x_2))$ is of the form $c_X(t)\equiv \exp(it^T\Delta)\phi(t^T\Omega t)$, where $i=\sqrt{-1}$.\label{item: mult}
 \end{enumerate}
\end{ass}
Each of these classes of specifications is permissive in some dimensions, and restrictive in others. As exchangeability is satisfied for most commonly used families of copulas, \ref{item: additive} is very flexible in the choice of the distribution $H$, but allows neither strict super- nor strict submodularity of the production function. \ref{item: binary} provides full flexibility in the choice of the production function and the distribution of relative concerns, but at the cost of making skills binary.  \ref{item: mult} allows for both super- (for $c>0$) and submodular ($c<0$) production functions, and for  many standard traits distributions, such as the log-normal and log-t-student, but at the cost of making very specific functional form assumptions.

\begin{prop}\label{prop: assumptions}
Assumption \ref{ass: suff} implies Assumptions \ref{ass: common} and \ref{ass: copula}.
\end{prop}

\deferred[proof: assumptions]{\paragraph*{Proof of Proposition \ref{prop: assumptions}}

I will show that each of \ref{item: additive}-\ref{item: binary} implies Assumptions \ref{ass: common} and \ref{ass: copula}.\looseness=-1

\ref{item: additive}: Assumption \ref{ass: common} is satisfied trivially, with $v_1(x_1^j, x_2^j; x_1, x_1')=1-G_{x_2}(x_2)$. The copula of $x_1, v_1$ and $x_1, 1/(1+2x_2)$ are the same, and thus Assumption \ref{ass: copula} is satisfied as well. \looseness=-1
    \ref{item: binary}: Assumption \ref{ass: suff} $\Rightarrow$ Assumption \ref{ass: common} immediately, because if the distribution of skill is binary, then there are only two levels of skill.

    I will show that \ref{ass: suff} $\Rightarrow$ Assumption \ref{ass: copula} by constructing copula $C_h$ of $x_1, v_1$ that satisfies  Assumption \ref{ass: copula}. First, define the functions
\begin{align*}
    &G(v) \equiv \text{Pr}(v_1 \leq v | X_2=L), \quad Z(v) \equiv \text{Pr}(v_1 \leq v | X_2=H)=2v-G(v),\\
    &C^1(u, v)=\frac{G(\min\{u, 0.5\})G(\min\{v, 0.5\})}{2G(0.5)}, \\ &C^2(u, v) = \frac{G(\max\{u, 0.5\})-G(0.5)}{2}\frac{Z(\min\{v, 0.5\})}{Z(0.5)}, \\
    &C^3(u, v)=\frac{(Z(\max\{u, 0.5\})-Z(0.5))(Z(\max\{v, 0.5\})-Z(0.5)}{2(1-Z(0.5))}.
\end{align*}

The candidate copula is then
\begin{equation}\label{eq: symmcopula}
  C_h(v_1, v_1) \equiv C^1(v_1, v_1)+(C^2(v_1, v_1)+C^2(v_1, v_1))+C^3(v_1, v_1).
\end{equation}
$C_h$ is symmetric because $C^1(u, v), C^3(u, v)$ and $C^2(u, v)+C^2(v, u)$ are all symmetric. It is a copula (a) because $\deriv{v_1}\deriv{v_1}C_h(v_1, v_1)$ exists almost everywhere, and is strictly positive wherever it exists, and (b) by symmetry and the facts that $C_h(v_1, 0)=0$, $C_h(v_1, 1)=v_1$. What remains to be shown is that $C_h(H_{x_1}(x_1), v_1)=\text{Pr}(X_1 \leq x_1, v_1 \leq v_1)$. This is true for $x_2=h$ by the definition of a copula, and follows for $x_2=l$ by inspection of  \eqref{eq: symmcopula} and the fact that $\text{Pr}(X_1 \leq l, v_1 \leq v_1)=0.5G(v)$.

\ref{item: mult}: An important property of the elliptical family of distributions, is that any linear transformation of an elliptically distributed random variable is elliptically distributed and has the same generator function $\phi$ \citep[e.g. Theorem 2.16 in ][]{Fang1990}. Formally, if $\mathbf{z} \sim E C_{n}(\Delta_z, \Omega_z, \phi)$ with $\operatorname{rank}(\Omega_z)=k, \mathbf{B}$ is an $m \times n$ matrix
and $\boldsymbol{p}$ is an $m \times 1$ vector, then
$
\boldsymbol{p}+\mathbf{B} \mathbf{z} \sim E C_{m}\left(\boldsymbol{p}+\mathbf{B} \Delta_z, \mathbf{B} \Omega_z \mathbf{B}^T, \phi\right).
$
Define the random variable $\bar{\mathbf{V}}=\mathbf{A} (\ln X_1, \ln (1+2X_2))^T$, where $A\equiv\left(\begin{smallmatrix} 1 & 0 \\ c & -1 \end{smallmatrix}\right)$. Clearly, $v_1(x_1^j, x_2^j; x_1, x_1')=\text{Pr}(\bar{v_1} \leq c \ln x_1-\ln(1+2x_2) )$. Assumption \ref{ass: common} is thus satisfied; Assumption \ref{ass: copula} is satisfied as well, because the copula of $(X_1, v_1)$ is also the copula of $\bar{\mathbf{V}}$, and copulas of log-elliptical random variables are exchangeable.
}

In Section \ref{sec: char}, I will derive general results that hold for, at least, all of the cases covered by Assumption \ref{ass: suff}. In Section \ref{sec: ToF}, I will focus on the binary skills case (\ref{item: binary}), as this case allows for both super- and submodular production functions while retaining tractability.

\subsection{Discussion}
\paragraph{Global Status}
The model features only within-firm social comparisons (local status) but no between-firm comparisons (global status). That is unrealistic: people compare themselves not only to their co-workers, but also to friends, family and acquaintances.
Fortunately, restricting attention to local status is without loss of generality. To see this, let us add a global status term into each worker's utility function:
\[u^h(w^k, w^j, \bar{F}; x_2^k, x_3^k)=w^k+x_2^k(w^k-w^j)+x_3^k(w^k+w^j-\bar{F}),\]
where $x_3\geq -0.5$ is worker-specific preference for global status, and $\bar{F}$ is the average economy-wide production. The idea here is simple: Social comparisons outside of the workplace are likely based on firm-wide characteristics, which are more easily observable for an outsider than the worker's individual wage. Thus, people who work for a firm that produces a lot of output, and thus also pays high wages on average, enjoy high global status.

As any pair of workers is of measure zero, the sorting decision of any individual worker has no impact on $x_3 \bar{F}$, and thus this term can be dropped.
We can then use the same monotone transformation of utility as in Section \ref{sec: TUR} to get
\[\tilde{u}^h(w^k, w^j; x_2^k, x_3^k)=w^k-w^j+(w^k+w^j)(1+2x_3^k)(1+2x_2^k)^{-1}.\]
In other words, we can define a new random variable $\tilde{x_2}\equiv 2(x_2-x_3)/(1+2x_2),$ and then the model with halo effect becomes isomorphic to the baseline model in which workers' type is $(x_1, \tilde{x_2})$. Thus, global status provides a compelling interpretation for negative strength of relative concerns: Workers with a negative $\tilde{x_2}$ are simply workers who care more strongly about their global status than their within-firm relative concerns.\footnote{Some of the results regarding wage inequality assume that $\underline{x}_2 \geq 0$: In the light of this discussion, this assumption requires that workers care more strongly about within- then between-firm relative concerns.}

\paragraph{Inequity Aversion}\label{par: inequity}  My model can be also used to study the impact that inequity aversion has on labour market sorting.
Consider a simple \cite{Fehr1999} utility function with heterogeneous inequity aversion:
\[U(w^k, w^s) = w^k - \alpha \max\{w^j-w^k,0\}-\beta\max\{w^k-w^j,0\}\]
Here,  $\alpha\geq 0$, $\beta \in [0, \frac{1}{2})$---to ensure that, keeping the team's output constant, utility increases in own wage---and both $\alpha$ and $\beta$ differ across individuals. In general, a sorting model in which workers have this utility function is different from mine. However, if skill are binary, then the inequity aversion model is isomorphic to my model if $x_2=-\beta$ for high skill workers and $x_2=\alpha$ for low-skill workers. To understand why, note that because utility increases in own wage and output increases in skill, in the relative concerns model, the high-skill workers will always earn more than the low-skill worker in any cross-match. Thus, having negative (positive) relative concerns is equivalent to having inequity aversion for high- (low-)skill workers.
I will call a preference distribution $H_{x_2}$ \emph{inequity aversion equivalent} if $\text{Pr}(X_2 \in (-0.5, 0] | x_1=H)=\text{Pr}(X_2 \geq 0 | x_1=L)=1$, in which case the binary skill model is isomorphic to the inequity aversion model. This class of preferences will play an import role in the discussions about sorting (Section \ref{sec: sorting}) and wage inequality (Section \ref{sec: inequality}). \looseness=-1

\paragraph{Relative Concerns as Private Information}\label{par: PI} My model, in which there is complete information about workers' types, is isomorphic to a model in which only skills are public information, but preferences are private information.  I show this formally in the Appendix, but the intuition is simple:  A worker's strength of relative concerns does not affect their co-worker's payoff; only the wage offered to the co-worker and the worker's skill do.
As a result, the co-worker is indifferent between all workers of the same skill who offer them the same wage, and
workers have no incentive to lie about the strength of their relative concerns.

\deferred[deriv: truthtelling]{\paragraph*{The Argument from \ref{par: PI}}
Truth-telling about ones preference is incentive compatible under the sorting and payoff functions that hold in the competitive equilibrium.
 To see this, suppose that $x_1$ is perfectly observable, but $x_2$ is not.  Workers first announce some $\hat{x_2}$, and after that sorting commences with $(x_1, \hat{x_2})$ treated as each workers true type. Finally, on top of the requirements specified in Definition \ref{defi: eqm}, we impose the \emph{truth-telling condition}:
\begin{equation}\label{eq: truth}
\begin{split}
   u^*(x_1, x_2)&=\max_{\hat{x_2}} (1+2x_2)w^*(x_1, \hat{x_2})-x_2F(x_1, \mu^*_1(x_1, \hat{x_2}))\\
  &=\max_{\hat{x_2}} (1+x_2)F(x_1, \mu^*_1(x_1, \hat{x_2}))-(1+2x_2)w^*(\mu(x_1, \hat{x_2}),
\end{split}
\end{equation}
where $w^*$ as defined in  \eqref{eq: wageutility}. In other words, worker $\mathbf{x}$ is free to match with any co-worker of a worker with skill $x_1$ as long as they pay them the same wage; and, of course, truth-telling requires that the utility maximising choice is the one that corresponds to their true $x_2$.
Critically, however, every worker was equally free to do so under complete information! Formally, we have that
\[
 \psi(\mathbf{x}^j, \mathbf{x}^k, u(\mathbf{x}^j))= (1+x_2^k)F(x_1^k, x_1^j)-(1+2x_2^k)w^*(\mathbf{x}^j)
\]
and thus individual rationality implies  \eqref{eq: truth}.
Overall, therefore, the requirement of truth-telling does not change the equilibrium conditions.}

\section{Characterising the Equilibrium}\label{sec: char}
In this Section, I will characterise the equilibrium sorting, payoff and wage functions, and use this characterisation to explore how the presence of relative concerns affects sorting patterns, welfare and wage inequality in comparison to the benchmark.

\subsection{Equilibrium Sorting}

It is well-established that in two-sided sorting problems with transferable utility the solution to the planner's problem coincides with the competitive equilibrium \citep{Gretsky1992}. Further, \cite{McCann2010} prove that the Monge-Kantorovich duality holds also for one-sided problems with transferable utility, from which it follows that the solution of the planner's problem coincides with the competitive equilibrium of my model's TU-representation.

\begin{theo}[\cite{McCann2010}]\label{theo: dualprimal}
In a sorting model with transferable utility, sorting $\mu^*$ and a payoff function $\tilde{u}^*$ constitute a competitive equilibrium if and only if $\mu^*$ solves the planner's problem:
\[\mu^* \in \argmax_{\mu \in \mathcal{S}(H)} V_{P\mathbf{x}}(\mu), \text{ where } V_P(\mu)\equiv\int_{I_{\mathbf{x}}} \Pi(\mathbf{x}, \mu(\mathbf{x})) \mathrm{d} H(\mathbf{x}),\] and $\tilde{u}^*$ solves its dual problem, that is $$\tilde{u}^* \in \argmin_{\tilde{u}} \left(\int_{I_{\mathbf{x}}} \tilde{u}(\mathbf{x}) \mathrm{d} H(\mathbf{x}) \quad \text{s.t. } \Forall_{(\mathbf{x}^j, \mathbf{x}^k)\in I_{\mathbf{x}}^2} \tilde{u}(\mathbf{x}^k)+\tilde{u}(\mathbf{x}^j) \geq \Pi(\mathbf{x}^j, \mathbf{x}^k)\right).$$
\end{theo}

Thus, I can characterise the equilibrium sorting of my model by solving the planner's problem for the TU-representation presented in Section \ref{sec: TUR}. Under Assumptions \ref{ass: common} and \ref{ass: copula}, the solution to the planner's problem exists and is easy to characterise.
\begin{theo}[Equilibrium Sorting] \label{theo: exuniq}
Under Assumptions \ref{ass: common} and \ref{ass: copula}, a sorting $\mu^*$ is an equilibrium sorting if and only if induces positive and assortative matching between workers' $x_1$ and co-workers' $v_1$, that is, iff $\mu^*$ satisfies  (a) $\text{Pr}(\mu^*_1(\mathbf{x}) \leq x)=H_{x_1}(x)$ and $\text{Pr}(v_1(\mu^*(\mathbf{x})) \leq v)=v$, and (b) $x_1^{\prime} > x_1 \Rightarrow v_1(\mu^*(\mathbf{x}^{\prime})) > v_1(\mu^*(\mathbf{x}))$.\newline
In particular, if $H_{x_1}$ is strictly increasing, then the equilibrium sorting is given by
\begin{equation}\label{eq: sorting}\mu^*(x_1, x_2)=[H_{x_1}^{-1}(v_1(x_1, x_2)), z(H_{x_1}^{-1}(v_1(x_1, x_2)), H_{x_1}(x_1) ]^T,\end{equation} where $z(x_1, v_1(x_1, x_2))\equiv x_2$ is the inverse of $v_1$ with respect to $x_2$.

\end{theo}
 \begin{proof}
 Define $\mathbf{v}(\mathbf{x}^k, \mathbf{x}^j)\equiv(v_1(\mathbf{x}^k), x_1^j)$, that is, a vector of the worker's skill-preference index and the co-worker's skill. The idea of the proof is to rewrite the planner's problem in terms of the vectors $\mathbf{v}(\mathbf{x}^k, \mathbf{x}^j), \mathbf{v}(\mathbf{x}^j, \mathbf{x}^k)$ and show that the resulting surplus function implies that $x_1^j$ is a complement to  $v_1(\mathbf{x}^k)$ and $x_1^k$ is a complement to  $v_1(\mathbf{x}^j)$, and that there is no other relevant complementarity or substitutability between the elements of  $\mathbf{v}(\mathbf{x}^k, \mathbf{x}^j), \mathbf{v}(\mathbf{x}^j, \mathbf{x}^k)$. Thus, the planner wants to match workers with high $x_1^k$ to co-workers with high $v_1^j$.\looseness=-1

 Formally, select an arbitrary $\tilde{x}_1 \in I_{x_1}$, and define
\begin{align*}
    &\pi(\mathbf{v}(\mathbf{x}^k, \mathbf{x}^j), \mathbf{v}(\mathbf{x}^j, \mathbf{x}^k))\equiv V^{-1}(v_1(\mathbf{x}^k); x_1^j, \tilde{x}_1)+V^{-1}(v_1(\mathbf{x}^j); x_1^k, \tilde{x}_1)\\
    &V_{P\pi}(\mu)\equiv \int_{I_{\mathbf{x}}}\pi(\mathbf{v}(\mu(\mathbf{x}), \mathbf{x}), \mathbf{v}(\mathbf{x}, \mu(\mathbf{x}))) \mathrm{d} H(\mathbf{x}).
\end{align*}

Note that
\[\Pi(\mathbf{x}^k, \mathbf{x}^j)=\pi(\mathbf{v}(\mathbf{x}^k, \mathbf{x}^j), \mathbf{v}(\mathbf{x}^j, \mathbf{x}^k))+(F(x_1^k, \tilde{x}_1))(1+2x_2^k)^{-1}+(F(\tilde{x}_1, x_1^j))/(1+2x_2^j).\]
Because the last two terms are additively separable in $\mathbf{x}^j, \mathbf{x}^k$, they can be added or subtracted from the surplus function with no impact on the maximiser of the planner's problem, so that
\[\max_{\mu \in \mathcal{S}(H)} V_{P\mathbf{x}}(\mu)=2E_H\left((F(x_1, \tilde{x}_1))/(1+2x_2) \right)+\max_{\mu \in \mathcal{S}(H)} V_{P\pi}(\mu).\]
Because $\mathbf{v}(\mathbf{x}^k, \mathbf{x}^j)$ depends only on $x_1^j$ and $v_1(\mathbf{x}^k)$, the strength of relative concerns $x_2$ affects a worker's match \emph{only through its impact on the rank $v_1(\mathbf{x}^k)$}.

 By construction, the mapping $\pi$ is additively separable in $\mathbf{v}(\mathbf{x}^k, \mathbf{x}^j)$ and $\mathbf{v}(\mathbf{x}^j, \mathbf{x}^k)$. Therefore, the only aspects of $\mu$ that affect $V_{P\mathbf{v}}(\mu_g)$ are the bi-variate distributions of  $\mathbf{v}(\mathbf{x}^k, \mathbf{x}^j)$ and $\mathbf{v}(\mathbf{x}^j, \mathbf{x}^k)$ it induces.  Note that
 $V^{-1}(v_1; v_2, \bar{x}_1)=(F(x_1, v_2)-F(x_1, \bar{x}_1))(1+2 z(x_1, v_1))^{-1}$ for all $x_1$ in the domain of $z(\cdot, v_2)$. It follows by  Assumption \ref{ass: common} that
 \[ \pi((v_1, v_2'), \mathbf{v}^{\prime \prime})-\pi((v_1, v_2), \mathbf{v}^{\prime \prime})=V^{-1}(v_1; v_2', \bar{x}_1)-V^{-1}(v_1; v_2, \bar{x}_1)=V^{-1}(v_1; v_2', v_1)\]
 increases in $v_1$ and thus $\pi(\mathbf{v}(\mathbf{x}^k, \mathbf{x}^j), \mathbf{v}(\mathbf{x}^j, \mathbf{x}^k))$ is supermodular in $\mathbf{v}(\mathbf{x}^k, \mathbf{x}^j)$.
Therefore, by standard results \citep[see, for example, Theorem 4.3 in][]{Galichon2016}, $V_{P\mathbf{v}}(\mu_g)$ cannot reach a value higher than that achieved for sortings that satisfy (b), that is, sorting which ensure that $v_1^j$ increases deterministically in $x_1^k$. Because $v_1(\mu^*(\mathbf{x}))$ strictly increases in $x_1$, $\mu^*_1(\mathbf{x})$ increases in $v_1(\mathbf{x})$; together with (a)---which must be satisfied for any feasible sorting---this implies that $\mu^*_1(\mathbf{x})$ depends on $x_1$ only through $v_1(\mathbf{x})$. It thus follows that the copula of $(\mu_1^*(\mathbf{x}), v_1(\mu^*(\mathbf{x}))$ is the same as the copula of $(v_1(\mathbf{x}), x_1)$, and thus $\mu^*$ is feasible by Assumption \ref{ass: copula}; it follows that $\mu^* \in\argmax_{\mu \in \mathcal{S}(H)} V_{P\mathbf{x}}(\mu)$. Finally, with strictly increasing $H_{x_1}$ only sorting $\mu^*$ satisfies (a) and (b).

\end{proof}

In equilibrium, high-skill workers match workers with high skill-preference index. To better understand what this implies for sorting in the $(x_1, x_2)$ space,  let us work through the case of homogeneous $x_2$, as well as the three specifications satisfying Assumption \ref{ass: suff}.

\paragraph{Homogeneous Relative Concerns}  If $x_2$ is the same for all workers, then $L(\mathbf{x}^k, x_1', x_1)$ depends on $x_1^k$ only. Assumption \ref{ass: common} is then satisfied if and only if the production function is either strictly supermodular or strictly submodular. Under strictly supermodular production $L(\mathbf{x}^k, x_1', x_1)$ strictly increases in $x_1^k$, and thus $v_1(\mathbf{x})=H_{x_1}(x_1)$ and we obtain positive and assortative matching (PAM) in skills; under strictly submodular production $L(\mathbf{x}^k, x_1', x_1)$ strictly decreases in $x_1^k$, and thus $v_1(\mathbf{x})=1-H_{x_1}(x_1)$ and negative assortative matching (NAM) obtains. Therefore, the sorting patterns are exactly the same as those derived by \cite{Becker1973} and \cite{Sattinger1979} for the model without relative concerns. This is our first insight: sorting depends on the strength relative concerns only if workers' preferences are heterogeneous.\looseness=-1

\paragraph{Additive Production}\label{par: sortingadd} Under Assumption \ref{item: additive}, $L(\mathbf{x}^k, x_1', x_1)$ becomes $(K(x_1')-K(x_1))(1+2x_2^k)^{-1}$ and thus depends only on $x_2^k$ (negatively) but not on $x_1^k$. Accordingly, $v_1(\mathbf{x})=1-H_{x_2}(x_2)$ and
\begin{equation}
        \mu^*(\mathbf{x})=[H_{x_1}^{-1}(1-H_{x_2}(x_2)), H_{x_2}^{-1}(1-H_{x_1}(x_2)). \label{eq: sortingadd}
    \end{equation}
When production is additive, any sorting pattern produces the same total output. Therefore, the only aspect of sorting that matters for the social planner (and thus also in equilibrium) is the social comparisons it induces.
The welfare-maximising sorting induces then negative and assortative matching between a worker's skill and the strength of their co-worker's relative concerns: A match between a high-skill worker with strong relative concerns and a low-skill worker with weak relative concerns allows the high-skill worker to enjoy his high status, without imposing much of a loss on the low-skill worker. High-skill workers with weak relative concerns and low-skill workers with strong relative concerns self-match.\looseness=-1

\paragraph{Binary Skills}\label{par: sortingbinary}  Under Assumption \ref{item: binary}, $x_1 \in \{l, h\}$ where $h> l$, and thus $L(\mathbf{x}^k; h, l)=(F(x_1^k,h)-F(x_1^k, l))(1+2x_2^k)^{-1}$. Denote the distribution of $x_2$ conditional on $x_1$ by $G_{x_1}$; it follows directly from the definition of $v_1(\mathbf{x})$ that
\begin{equation}\label{eq: v2binary}
    v_1(\mathbf{x})=\sum_{j \in \{l, h\}} 0.5\left[1-G_{j}\left((x_2+0.5)\frac {F(j, h)-F(j, l)}{F(x_1,h)-F(x_1, l) }-0.5\right) \right].
\end{equation} By Theorem \ref{theo: dualprimal} any worker with $v_1(\mathbf{x}^k) > (<) 0.5$ matches a co-worker of high (low) skill; a rearrangement of  \eqref{eq: v2binary} yields then that in equilibrium high-skill workers with $x_2 \geq G_h^{-1}(\bar{y})$ match low-skill workers with $x_2 \leq G_l^{-1}(\bar{y})$ and all remaining workers self-match, where $\bar{y}=1$ if $a_F < T_h(0)/T_l(0)$, $\bar{y}=0$ if $a_F>T_h(1)/T_l(1)$, and $\bar{y}$ solves $a_F=T_h(\bar{y})/T_l(\bar{y})$ otherwise, with
\[  a_F \equiv \frac{F(h, h)-F(h, l)}{F(h, l)-F(l, l)}, \qquad
    T_h(y)=G_h^{-1}\left(y  \right)+0.5, \qquad
    T_l(y)=G_l^{-1}\left(1-y \right)+0.5.
\]

With supermodular (submodular) production output maximisation requires positive (negative) and assortative sorting in skills. However, the need to maximise production needs to be traded-off against the desire to match high-skill workers with workers that care little about social comparisons.
The outcome of this trade-off depends in general on (a) how strong the supermodularity (submodularity) of the production function is and (b) how strong are the relative concerns of high-skill workers compared to low-skill workers. High-skill workers with very strong relative concerns always match low-skill workers with very weak relative concerns. However, the definition of ``strong'' or ``weak'' relative concerns depends very much on the complementarity between workers of the same skill, as captured by $a_F$.  In general, the stronger the complementarity, the larger the difference in the strength of relative concerns needs to be to warrant a match between a high- and a low-skill worker.

\paragraph{Multiplicative Production and Log-Elliptically Distributed Traits}\label{par: sortingmult} Under Assumption \ref{item: mult}, $L(\mathbf{x}^k, x_1', x_1)=(x_1^c-(x_1')^c)(x_1^k)^c(1+2x_2^k)^{-1}$; thus $v_1(\mathbf{x}^k)$ is simply equal to worker's $\mathbf{x}^k$ rank in the distribution of $\bar{v}_2\equiv c \ln x_1 - \ln(1+2x_2)$. As any linear transformation of an elliptically distributed random variable remains elliptically distributed with the same generator function $\phi$ \citep[e.g. Theorem 2.16 in ][]{Fang1990}, it follows that $(x_1, \bar{v}_2) \sim EC_2(\mathbf{A} \Delta, \mathbf{A} \Omega \mathbf{A}^T; \phi)$, where $A\equiv\left(\begin{smallmatrix} 1 & 0 \\ c & -1 \end{smallmatrix}\right)$. Denote the  square root of the ratio of variances of $\bar{v}_2$ and $x_1$ by $r$:  it follows then from Theorem \ref{theo: dualprimal} and some linear algebra that:
\begin{equation}\label{eq: sortingmult}\mu_1^*(\mathbf{x})=\left(\frac{x_1^c}{1+2x_2}\right)^{1/r}e^{\delta_1(1-\frac{c}{r})+\frac{\delta_2}{r}} , \quad \mu_2^*(\mathbf{x})=\frac{e^{(1+\frac{c}{r})[\delta_1(1-\frac{c}{r})+\delta_2]}}{2}\left(\frac{x_1^{c-\frac{r^2}{c}}}{1+2x_2}\right)^{c/r}-0.5.\end{equation}

As in the binary case, if $c>0$ ($c<0$), and thus if production is supermodular (submodular), then the equilibrium sorting optimally trades-off the need to match high-skill workers to high- (low-) skill co-workers, with the need to match high-skill workers to workers who have weak relative concerns.
There is, however, one noteworthy change compared to the binary skills case: When skills are elliptically distributed, then the measure of self-matching workers is 0 as long as $\Omega$ is of full rank. This is because, instead of self-matching, a worker with very high skill and weak (but not very weak) relative concerns can now match a worker with high (but not very high) skill and very weak relative concerns.
\subsubsection{Sorting in Skills}\label{sec: sorting}
The preceding discussion made it very clear that the direction and the strength of sorting in the skill dimension depends not only on the properties of the production function, but also on the distribution of relative concerns. Indeed, in the additive case sorting in skills is solely determined by the interdependence between skill and the strength of relative concerns---for example, sorting is perfectly positive and assortative in skills if and only if $x_2$ decreases deterministically in $x_1$. More generally, whether workers of high skill are complements or substitutes matters for the direction and the strength of sorting in skill as well, just as it does in the standard sorting model---however, even then the distribution of relative concerns continues to play a (possibly dominant) role.

\begin{prop}\label{prop: sorting}
Consider an economy $(F, H)$ which satisfies Assumption \ref{ass: suff}. For any $\rho \in [-1, 1]$  there exists a distribution of traits $H$, such that (a) economy $(F, \tilde{H})$ satisfies Assumption \ref{ass: suff}, (b) the marginal distribution of skill is the same under $H$ and $\tilde{H}$ ($H_{x_1}=\tilde{H}_{x_1}$), and (c) $\text{Corr}(X_1, \mu^{*}_1(X_1, X_2))=\rho$ for every equilibrium sorting function $\mu^*$ of economy $(F, \tilde{H})$.
\end{prop}

\deferred[proof: propsorting]{\paragraph*{Proof of Proposition \ref{prop: sorting}}
It suffices to show that there exists a parameterised family of joint distributions $\tilde{H}^\rho$ that (i) satisfies conditions (a) and (b) and (ii) induces a family $C^\rho$ of copulas of the joint distributions of $(X_1, \mu^*_1(X_1, X_2))$ which depends continuously on the parameter $\rho$ and nests both Frechet–Hoeffding upper and lower bounds.\footnote{In the binary skills case, this suffices because the correlation between workers' and their co-workers' skill depends on $\bar{y}$ only, which is the same for all possible equilibrium sortings.} This is because $X_1$ and $\mu^*_1(X_1, X_2)$ have the same marginal distributions, and thus the Frechet–Hoeffding upper (lower) bound copula produces    $\text{Corr}(X_1, \mu^*_1(X_1, X_2))$ equal to $1$ ($-1$).

If $(F, H)$ satisfies Assumption \ref{item: additive}, then $\text{Corr}(X_1, \mu^*_1(X_1, X_2))=\text{Corr}(X_1, H_{x_1}^{-1}(1-H_{x_2}(X_2)))$, and the result follows by setting $\tilde{H}^\rho$  to be the  family of Gaussian copulas. If $(F, H)$ satisfies Assumption \ref{item: binary}, let us set $\tilde{H}^\rho$ so that $T_h(y)=1+y+(1-\rho)a_F$ and $T_l(y)=2-y+(1+\rho)/a_F$, which is clearly continuous in $\rho$ and,  by the discussion in \ref{par: sortingbinary}, attains the Frechet–Hoeffding lower (upper) bound for $\rho=-1$ ($\rho=1$).

Finally, if $(F, H$ satisfies Assumption \ref{item: mult}, then set $\tilde{H}^\rho$  to be $EC_2(\Delta, \Omega(\rho); \phi)$ distributed, where $\omega(\rho)_{11}=\omega_{11}$, $\omega(\rho)_{22}=4c^2\omega_{11}$, $\omega(\rho)_{12}=\omega(\rho)_{21}=\rho \sqrt{\omega(\rho)_{11}\omega(\rho)_{22}}$, and $\rho \in [-1, 1]$. It follows from the proof of Proposition \ref{prop: assumptions} that  $\ln X_1, \ln \mu^*_1(X_1, X_2)$ is  $EC_2$ distributed, with both marginals equal to $H_{x_1}$ and
$$\text{Corr}(\ln X_1, \ln \mu^*_1(X_1, X_2))=\text{sgn}(c)\frac{1-2 \rho \text{sgn}(c)}{\sqrt{(1-2\text{sgn}(c))^2+4\text{sgn}(c)(1-\rho)}}.$$ Therefore, the copula of $X_1, \mu^*_1(X_1, X_2)$ depends continuously on $\rho$, and for $\rho=1$ ($\rho=-1$) $\text{Corr}(\ln X_1, \ln \mu^*_1(X_1, X_2))=-1 (=1)$ and thus the copula of $X_1, \mu^*_1(X_1, X_2)$ reaches the Frechet–Hoeffding lower (upper) bound.
}

Proposition \ref{prop: sorting} implies that we can fix the production side of the economy---that is, the production function and the marginal distribution of skill---and yet produce any degree of sorting in skill in equilibrium, simply by altering the preference structure of the economy.  This means, in particular, that the strongly positive empirical correlation between co-workers' skills \citep[see, e.g.][]{Freund2022} is consistent with a strictly submodular production function.
The intuition is straightforward. Suppose that workers' skill and relative concerns are perfectly negatively correlated, which is the case, for example, if preferences are inequity aversion equivalent. A high-skill worker faces then a trade-off between maximizing their own wage by matching a low-skill worker, and minimising the within-firm wage differential by self-matching. If the relative concerns of high-skill workers are very weak compared to low-skill workers, then the welfare gain from self-matching outweighs the welfare loss stemming from the loss of output, and  positive assortative matching in the skill dimension prevails. \looseness=-1

\subsection{Equilibrium Payoffs and Wages}

Let us start the derivation of equilibrium payoffs by rewriting \eqref{eq: max} as

\[\tilde{u}^*(\mathbf{x}^k)= \max_{\mathbf{x}^j \in I_{\mathbf{x}}}\Pi(\mathbf{x}^j, \mathbf{x}^k)-\tilde{u}(\mathbf{x}^j),\]
so that  \eqref{eq: Pi} and the Envelope Theorem imply
\[\deriv{x_2}\tilde{u}^*(\mathbf{x})=-2F(x_1, \mu^*_1(\mathbf{x}))(1+2x_2)^{-2}.\]

Therefore, the equilibrium payoff function must satisfy
\begin{equation}\label{eq: eqmutility}
    u^*(\mathbf{x})=0.5(1+2x_2)\left(\tilde{u}^*(x_1, x_2^*)+2\int_{x_2^*}^{x_2}F(x_1, \mu^*_1(x_1, s)) \, \mathrm{d} (1+2s)^{-1}\right),
\end{equation}
Thus, as long as for every $x_1$ there exists some $x_2^*$ for which  $\tilde{u}^*(x_1, x_2^*)$ can be determined, we would be able to derive $u^*(\mathbf{x})$. %
The obvious candidate for such $(x_1, x_2^*)$ are workers who self-match in equilibrium, that is match with a co-worker of the same skill. As $\mu^*(\mu^*(\mathbf{x}))=\mathbf{x}$, co-workers of the same skill must split output equally---otherwise one of them could do strictly better by matching a worker of identical type, rather than just skill---and thus
the transformed utility of the self-matching worker $(x_1, x_2^*(x_1)$ equals
$F(x_1, x_1)/(1+2x_2^*(x_1))$. Therefore, the equilibrium utility function can be readily derived from  \ref{eq: eqmutility} as long as for every $x_1$ there exists some $x_2^*(x_1)$ such that the worker $(x_1, x_2^*(x_2))$ self-matches. While I will shy away from fully characterising the necessary and sufficient conditions for this to occur, it follows from \ref{par: sortingadd}-\ref{par: sortingmult} that this is generically the case as long as Assumption \ref{ass: suff} is satisfied.\footnote{This is  the case as long as the sorting function is not perfectly negative and assortative, that is, as long as $\bar{y}<1$ under Assumption \ref{item: binary} and $x_1, \bar{v}_2$ are not perfectly negatively correlated under Assumptions \ref{item: additive} and \ref{item: mult}.\label{foot: generic}}\looseness=-1
\begin{prop}[Equilibrium Payoffs and Wages]\label{prop: utilandwage}
    If Assumptions \ref{ass: common} and \ref{ass: copula} are satisfied, and $\mu^*$ is such that for every $x_1$ there exists a $x_2^*(x_1)$ for which $\mu_1(x_1, x_2^*(x_1))=x_1$,
     then the equilibrium payoff function and wage functions $u^*, w^*$ are as follows:
     \begin{IEEEeqnarray}{rCl}
    u^*(\mathbf{x})&=&0.5F(\mathbf{x_1})+\int^{x_2}_{x_2^*(x_1)}(s-x_2)(1+2s)^{-1} \mathrm{d} F(x_1,\mu^*_1(x_1, s)) \label{eq: eqmutility2} \\
    w^*(\mathbf{x})&=&0.5F(\mathbf{x_1})+\int^{x_2}_{x_2^*(x_1)}s(1+2s)^{-1} \mathrm{d} F(x_1,\mu^*_1(x_1, s)). \label{eq: wage2}
\end{IEEEeqnarray}
Here, $\mathbf{x_1}$ denotes $(x_1, x_1)$.
\end{prop}
\begin{proof}
    \eqref{eq: eqmutility2} follows from substituting $\tilde{u}^*(x_1, x_2^*)=F(x_1, x_1)/(1+2x_2^*)$ and \eqref{eq: sorting} into  \eqref{eq: eqmutility}, integration by parts and some rearranging. \eqref{eq: wage2} follows then from \eqref{eq: eqmutility2} and
    \begin{equation}
        w^*(\mathbf{x})=0.5\left ((2u^*(\mathbf{x}))/(1+2x_2)+(1-1/(1+2 x_2))F(x_1, \mu_1^*(\mathbf{x}))\right), \label{eq: wageutility}
    \end{equation}
    which itself is implied by  \eqref{eq: utildefinition} and $F(x_1, \mu_1^*(\mathbf{x}))=w(\mathbf{x})+w(\mu^*(\mathbf{x})$.
\end{proof}

Let me draw your attention to two properties of the equilibrium payoff and wage functions.

\paragraph{Wages and Relative Concerns} Keeping skill constant, workers with stronger relative concerns earn lower wages: as $v_1(\mathbf{x})$ decreases in $x_2$, so does $\mu_1^*(\mathbf{x})$ and thus
\[w(x_1, x_2')-w(x_1, x_2)= \int^{x_2'}_{x_2}s(1+2s)^{-1} \mathrm{d} F(x_1,\mu^*_1(x_1, s))< 0,\]
as long as $x_2'>x_2>0$. This implies that workers with strong and positive relative concerns earn lower wages than workers with weak (but still positive!) relative concerns. The reason for this, perhaps, slightly counter-intuitive result is that \emph{there is no effort provision in this model} and hence ones \emph{relative} wage can be increased only by matching a less-skilled (and thus lower earning!) co-worker.
Alas, as production increases in skill, matching a less skilled co-worker comes at the cost of decreasing the worker's absolute wage.

\paragraph{A Trickle-Down Effect?} Proposition \ref{prop: utilandwage} suggests the presence of a \emph{trickle-down effect}. This effect is clearest when production is additive (Assumption \ref{item: additive}).
With binary skill the inverse function $z$ becomes $z(x_1, v_1)=H_{x_2}^{-1}(1-v_1)$.  Denote $H_{x_2}^{-1}(H_{x_1}(x_1))$ by $L(x_1)$, then \eqref{eq: sortingadd}, \eqref{eq: eqmutility2} and \eqref{eq: wage2} and integration by substitution yield:\looseness=-1
\begin{IEEEeqnarray}{rCl}
   u^*(\mathbf{x})&=&0.5F(\mathbf{x_1})+\int_{x_1}^{L^{-1}(x_2)}(L(s)-x_2)K'(s)(1+2L(s))^{-1} \label{eq: payoffadd} \\
    w^*(\mathbf{x})&=&0.5F(\mathbf{x_1})+\int_{x_1}^{L^{-1}(x_2)}L(s)K'(s)(1+2L(s))^{-1} \, \mathrm{d} s. \label{eq: wageadd}
\end{IEEEeqnarray}
 Consider a change in $K$ that increases $K'(x_1)$ for all $x_1$ above some cutoff $\hat{x}$, but leaves it unchanged otherwise. Note that $L(\cdot)$ does not depend on $K(\cdot)$; moreover, because $L$ is decreasing, $L(s)>x_2$ for $s \in [x_1, L^{-1}(x_1))$. Therefore, any such change in $K$ raises the payoffs of  low-skilled workers with weak relative concerns, as those workers have co-workers with skill above the cutoff. Importantly, however, low-skill workers with strong relative concerns are matched to co-workers with skill below the cutoff, and thus do not see any increase in payoff. Thus, in the additive production case, the gains from increased productivity trickle down, but only to those workers who do not have strong relative concerns.\looseness=-1

If production is not additive, the impact of changes in the productivity of high-skill workers on low-skill workers' wages and payoffs depends on its impact on the sorting pattern. This is easiest to see in the case of binary skills, where \eqref{eq: eqmutility2} and \eqref{eq: wage2} simplify to
\begin{IEEEeqnarray}{rCl}
    u^*(\mathbf{x})&=& 0.5\left(F(\mathbf{x_1})+\left(F(\mu_1^*(\mathbf{x}), x_1)-F(\mathbf{x_1})\right)\left(1-(1+2x_2)/(2T_{x_1}(\bar{y}))\right)\right) \label{eq: payoffbinary} \\
    w^*(\mathbf{x})&=&0.5\left(F(\mathbf{x_1})+\left(F(\mu_1^*(\mathbf{x}), x_1)-F(\mathbf{x_1})\right)\left(1-1/(2T_{x_1}(\bar{y}))\right)\right). \label{eq: wagebinary}
\end{IEEEeqnarray}
Here, $\mu_1^*(\mathbf{x})$ is as specified in \ref{par: sortingbinary}. Consider an increase in $h$, which necessitates an increase in $F(h, l)-F(\mathbf{l})$, but has an ambiguous impact on $F(\mathbf{h})-F(h, l)$ and $a_F$. If $a_F$ decreases or remains constant, then $\bar{y}$ must decrease and there are \emph{more} low-skill workers matched with high-skill workers than before. As the self-match payoff is unchanged for low-skill workers---and any worker has always the option of receiving the self-match payoff---it follows from revealed preference that all low-skill workers who change their matches are better off. Similarly, by inspection of \eqref{eq: payoffbinary}, if $a_F$ decreases than an increase in $h$ must make any low-skill worker who was remained matched with a high-skill worker better off.

If, instead, the increase in $h$ raises $a_F$, then an improvement in $h$ makes at least some low-skill workers worse off.\footnote{An example of such an increase in $h$ would be skill-biased technological change (Definition \ref{defi: sbtc} below) under submodular production function.} An increase in $a_F$ would increase $\bar{y}$ and thus also the number of self-matching low-skill workers; any newly self-matched low-skill worker will be worse off by revealed preference. Indeed, if the increase in $H$ increased only $F(\mathbf{h})$ with no or very weak effect on $F(h, l)$, then all low-skill workers would become worse off, and the trickle-down effect would disappear.

Overall,  any increase in high-skill workers' productivity that weakens sorting in skills benefits those low-skill workers who end up matched to high-skill workers (leaving other low-skill workers unaffected). Importantly, note that if $\underline{x}_2 \geq 0$ then this trickle-down effect increases not only the payoffs but also wages of the affected low-skill workers. To see why, note that $U(w^k, \mathbf{x}^, x_1^j)=(1+2x_2^k)w^k-x_2^kF(x_1^k, x_1^j)$ and thus without a change in wages, an improvement in the co-worker's productivity makes the worker worse off. Therefore, for the low-skill workers to become better off, their wage must increase.\looseness=-1

\subsubsection{Wage Inequality}\label{sec: inequality}
Intuitively, as long as all workers have positive relative concerns, wage inequality should be lower in the presence of relative concerns than if all workers received the self-match wage, simply because any low-skill workers who are not self-matched must earn higher---and any high-skill workers lower---wages than the self-match wage. Proposition \ref{prop: var} confirms this intuition for any specification of the model which satisfies Assumption \ref{ass: suff}---with the obvious caveat that under Assumption \ref{item: mult}, $\bar{x}_2=-0.5$, and thus the condition that $\bar{x}_2 \geq 0$ is replaced by the condition that there are very few workers with $x_2<0$.\footnote{This intuition could fail if skills were binary and $p>0.5$, because the presence of relative concerns would increase the inequality between the two sub-economies discussed in footnote \ref{foot: pneqhalf}.}

\begin{prop}\label{prop: var}
If either (i) $\underline{x}_2 \geq 0$ and one of Assumptions \ref{item: additive} or \ref{item: additive}  is satisfied or (ii) $H_{x_2}(0)$ is sufficiently close to 0 and Assumption \ref{item: mult} is satisfied, then $\text{Var}(w_S)\geq\text{Var}(w^*)$.\looseness=-1
\end{prop}
\deferred[proof: propvar]{\paragraph*{Proof of Proposition \ref{prop: var}}

\textbf{Assumption \ref{item: additive}}.
Notice that, by \eqref{eq: wageadd}, (a) the wage function and the average wage in the economy depend only on the marginals of the traits distribution, but not its copula and (b) if the copula of $H$ is the Frechet–Hoeffding lower bound, then each worker receives exactly half of the production of a self-matched team, which is the benchmark wage. Thus,  wage variance is lower in my model than in the baseline as long as $ \int_{I_\mathbf{x}} w(\mathbf{x})^2 \mathrm{d} H(\mathbf{x}) < \int_{I_\mathbf{x}} w(\mathbf{x})^2 \mathrm{d} \underline{C}(H_{x_1}(x_1), H_{x_2}(x_2)),$ where $\underline{C}(\mathbf{v})\equiv \max\{v_1+v_1-1,0\}$ is the Frechet–Hoeffding lower bound copula. This is clearly true---by the definitions of the  Frechet–Hoeffding lower bound and the supermodular order---because the square function is convex and  $w(\mathbf{x})$ is additevely separable and (under the assumption that $\underline{x_2} \geq 0$) increases in both variables, and thus $w(\mathbf{x})^2$ is supermodular.

\textbf{Assumption \ref{item: binary}}.
If $\bar{y}=1$ then $w_S=w^*$ and the result is immediate. If $\bar{y}=0$ and $\underline{x}_2 \geq 0$ then all low- (high) skill workers must earn more (less) than $w_S(l)$ ($w_S(h)$) and the result follows as well.
If $\tilde{y} \in (0, 1)$, then it follows from \eqref{eq: wagebinary} and $a_F=T_h(\bar{y})/T_l(\bar{y})$ that $(T_h(\bar{y})+T_l(\bar{y})) \Delta w^n=\Delta w^S$, and  $F(h,l)-(F(\mathbf{h})+F(\mathbf{l}))/2=\Delta w^S (T_l(\tilde{y})-T_h(\tilde{y}))/(T_l(\tilde{y})-T_h(\tilde{y}))$, where $\Delta w^S$ is the difference between the high- and low-skill self-matched wages, and $\Delta w^n$ is the difference between the high- and low-skill wages of workers who cross-match. Clearly then,  $\text{Var}(W_B)=0.25 (\Delta w^S)^2$ and
\[\text{Var}(W)=\text{Var}(W_B)\frac{1-\bar{y}+(1-\bar{y})\bar{y}(T_l(\bar{y})-T_h(\bar{y}))^2+\bar{y}(T_l(\bar{y})+T_h(\bar{y}))^2}{(T_l(\bar{y})+T_h(\bar{y}))^2},\]
Elementary algebra reveals that if $4T_l(\bar{y})T_h(\bar{y})\geq 1$ then $\text{Var}(W)-\text{Var}(W_B) \leq 0$. The result follows because $(0.5+\underline{x}_2)^2 \leq T_l(\bar{y})T_h(\bar{y})$.

\textbf{Assumption \ref{item: mult}}.
Consider an arbitrary wage function $w$ and a feasible sorting function $\mu$. Variance decomposition yields
$\text{Var}(w(x_1, x_2))= \text{BWI}(w, \mu)+\text{WWI}(w, \mu),$
where
\[\text{BWI}(w, \mu)=\text{Var}\left(w(\mathbf{x})+w(\mu(\mathbf{x})\right), \qquad \text{WWI}(w, \mu)=0.25\text{E}\left(w(\mathbf{x})-w(\mu(\mathbf{x}), \mu_2(\mathbf{x}))\right)^2; \]
thus it suffices to show that $\text{WWI}(w_B, \mu^*)\geq\text{WWI}(w^*, \mu^*)$ and $\text{BWI}(w_B, \mu^*)\geq\text{BWI}(w^*, \mu^*)$.

The first part is easy. It follows directly from  \eqref{eq: wage2} that worker $\mathbf{x}$ earns more (less) than their benchmark wage if $\mu_1^*(\mathbf{x}) \geq (\leq) \mathbf{x}$ and $z(x_1, H_{x_1}(s)) \geq 0$ for all $s \in [x_1, \mu_1(\mathbf{x}]$.
Note that $z(x_1, H_{x_1}(\mu^*(\mathbf{x}))=z(x_1, v_1(\mathbf{x}))=x_2$, and $x_1=\mu^*_1(\mu_1(\mathbf{x}), \mu^*_2(\mathbf{x}))$, and thus $z(x_1, H_{x_1}(x_1))=\mu_2(x_1, x_2)$. It follows, therefore, that if $\min\{x_2, \mu^*_2(x_1, x_2)\}>0$ and $\mu^*_1(\mathbf{x}) \geq (\leq) \mathbf{x}$ then $w(\mu^*(\mathbf{x}))\leq w_B(\mu^*(\mathbf{x}))$ and $w(\mathbf{x})\geq w_B(\mathbf{x})$. Keeping all other parameters constant, if $\delta_1$ becomes arbitrarily large compared to $\omega_{11}$ then the probability of $\min\{x_2, \mu_2(x_1, x_2)\}<0$ becomes arbitrarily small, and $\text{WWI}(w_B, \mu^*)\geq\text{WWI}(w^*, \mu^*)$.

Moving on to $\text{BWI}$, we have that
\begin{align*}
  c^2\text{BWI}(w_B, \mu^*)&=0.25\text{Var}(x_1^c+\mu_1^*(\mathbf{x})^c)=0.5(\text{Var}(x_1^c)+\text{Cov}(x_1^c, \mu_1^*(\mathbf{x})^c)),  \\
  c^2\text{BWI}(w^*, \mu^*)&=\text{Var}((x_1)^{c}\mu_1^*(\mathbf{x})^c)=E((x_1^{c}\mu_1^*(\mathbf{x})^c)^2)-E(x_1^{c}\mu_1^*(\mathbf{x})^c)^2.
\end{align*}
Define the random variables $s\equiv c(\ln x_1+ \ln \mu_1^*(\mathbf{x}))$ and $z\equiv(s-2\delta_1)/\alpha$; where $\alpha^2\equiv \text{Corr}(\ln x_1, \ln \mu_1^*(\mathbf{x}))+1$. By \citep[e.g. Theorem 2.16 in ][]{Fang1990}, $z \sim  EC_1(0, c^2\omega_{11}; \phi)$; denote the cdf of $z$ by $G$.  Next, let us write
\begin{IEEEeqnarray*}{rCl}
    c^2(\text{BWI}(w^*, \mu^*)-\text{BWI}(w_S, \mu^*))&=& e^{-4\delta_1} (\underbrace{0.5E(z^2)-E(z)^2}_{\equiv T(\alpha)})-A.
\end{IEEEeqnarray*}
Here, $A=0.5\text{Var}(x_1^c)+0.5\text{E}(x_1^c)\text{E}(\mu_1^*(\mathbf{x})^c))$. It is easy to see that if $\alpha=2$, then
$\text{BWI}(w^*, \mu^*)=\text{BWI}(w_S, \mu^*)$, and if $\alpha=0$, then $\text{BWI}(w^*, \mu^*) \leq \text{BWI}(w_S, \mu^*)$.\footnote{If $\alpha=0$, then $\text{Corr}(\ln x_1, \ln \mu_1^*(\mathbf{x}))$ and the variance of $s$ is 0.} Thus, by standard arguments, if $T(\cdot)$ is convex then $\text{BWI}(w^*, \mu^*)-\text{BWI}(w_S, \mu^*)<0$.

\textit{Proof that $T(\cdot)$ is convex.} Let us start by rewriting $T(\alpha)$
\begin{IEEEeqnarray*}{rCl}
    T(\alpha)&=&0.5\int_{-\infty}^{\infty} e^{2(\alpha z)}\mathrm{d} G(z)-\left(\int_{-\infty}^{\infty} e^{\alpha z}\mathrm{d} z\right)^2 \\
    &=&\int_{-\infty}^{\infty} e^{\alpha z}\left(\int_{-\infty}^{r} 0.5e^{\alpha z} -e^{\alpha r} \, \mathrm{d} G(r)+\int_r^{\infty} 0.5e^{\alpha z} -e^{\alpha r} \, \mathrm{d} G(r)\right) \, \mathrm{d} G(z)\\
      &=& \int_{-\infty}^{\infty} \int_r^{\infty} e^{\alpha r}(0.5e^{\alpha r} -e^{\alpha z})+ e^{\alpha z}(0.5e^{\alpha z} -e^{\alpha r}) \, \mathrm{d} G(r) \, \mathrm{d} G(z)\\
      &=&0.5\int_{-\infty}^{\infty} \int_{-\infty}^{\infty} 0.5e^{2\alpha z} -2e^{\alpha(r+z)}+0.5 e^{2\alpha r} \, \mathrm{d} G(r) \, \mathrm{d} G(z).
\end{IEEEeqnarray*}
Denote $0.5e^{2\alpha z} -2e^{\alpha(r+z)}+0.5 e^{2\alpha r}$ by $p(z, r; \alpha)$. Note that (a) $G(z)=1-G(z)$, because any elliptical distribution is symmetric and (b) $p(z, r;\alpha)=p(r, z;\alpha)$; we can thus write:
\[T(\alpha)=\int_0^{\infty}\int_0^{\infty} \underbrace{p(z, r; \alpha)+p(-z, r; \alpha)+p( z, -r; \alpha)+p(-z, -r; \alpha)}_{\equiv P(z, r; \alpha)} \, \mathrm{d} G(r) \, \mathrm{d} G(z).\]
Clearly, it suffices thus to show that $\derivsec{\alpha}P(z, r; \cdot) \geq 0$ for any $(z, r) \in \mathbb{R}^2_{+}$; as $P(z, r; \alpha)$ is symmetric in $z, r$, we can assume, wlog, that $z>r$.
First, note that
\begin{IEEEeqnarray*}{rCl}
\derivsec{\alpha}p(z, r; \alpha)&=&2(z^2e^{2 \alpha z}+r^2e^{2 \alpha r}-(r+z)^2e^{\alpha(r+z)}),\\
\deriv{z}\derivsec{\alpha}p(z, r; \alpha)&=&4(\underbrace{ze^{2 \alpha z}-(r+z)e^{\alpha(r+z)}}_{l(z,r; \alpha)})\\
&&+2\alpha(\underbrace{2z^2e^{2 \alpha z}-(r+z)^2e^{\alpha(r+z)}}_{k(z, r; \alpha)},\\
\derivsec{z}\derivsec{\alpha}p(z, r; \alpha)&=&4(e^{2 \alpha z}-e^{\alpha(r+z)})+2\alpha \deriv{z}\derivsec{\alpha}p(z, r; \alpha)\\&&+2 \alpha(4ze^{2 \alpha z}+\alpha(r+z)^2e^{\alpha(r+z)}).
\end{IEEEeqnarray*}

Next, note that
\begin{align*}
    &\deriv{z}\derivsec{\alpha}P(z, r; \alpha) \geq 0 \Rightarrow \derivsec{z}\derivsec{\alpha}P(z, r; \alpha) > 0,\\
    &K(z, r ;\alpha)\equiv k(z, r; \alpha)+k(-z, r; \alpha)+k( z, -r; \alpha)+k(-z, -r; \alpha) \geq \derivsec{\alpha}P(z, r; \alpha),
\end{align*}
because $z(e^{2 \alpha z}-e^{-2 \alpha z}>0$ and
\begin{equation*}\label{eq: diff}e^{2 \alpha z}-e^{\alpha(r+z)}+e^{-2 \alpha z}-e^{-z\alpha(z+r)}= \alpha \int_{r+z}^z e^{\alpha s}-e^{-\alpha s} \mathrm{d} s >0\end{equation*}
for any $r \in \mathbb{R}$ and $z \geq \max\{0, r\}$.
Finally, it is immediate that $\derivsec{\alpha}P(z, z; \alpha)=0$
and that $L(z, z; \alpha)=0$, where $L(z, r ;\alpha)\equiv l(z, r; \alpha)+l(-z, r; \alpha)+l( z, -r; \alpha)+l(-z, -r; \alpha)$. Jointly these facts imply that $\derivsec{\alpha}P(z, r; \alpha) \geq 0$ for all $(z, r) \in \mathbb{R}^2_{+}$. Suppose not. Then there must exist some $(z^*, r*) \in \mathbb{R}^2_{+}$ such that $\derivsec{\alpha}P(z^*, r*; \cdot) < 0$ and (by symmetry) $z^*>r*$. This is only possible if the set $\Omega\equiv \{ z \in [r*, z^*]: \deriv{z}\derivsec{\alpha}P(z, r; \alpha)<0\}$ is non-empty; denote its infimum by $z'$. For all $z \in [r^*, z']$ it must be the case that $\derivsec{\alpha}P(z, r^*; \alpha), \deriv{z}\derivsec{\alpha}P(z, r^*; \alpha) \geq 0$. This implies that $\deriv{z}\derivsec{\alpha}P(z', r^*; \alpha) > 0$; contradiction!

}

It follows immediately, that---under the conditions imposed by Proposition \ref{prop: var}---wage inequality in my model must be lower than in the benchmark as long as the benchmark wages are more unequal than the self-match wages, that is, as long as $\text{Var}(w_B) \geq \text{Var}(w_S)$.
 By standard results, this is always the case if production is either supermodular or additive, as then $w_B(\mathbf{x})=w_S(\mathbf{x})$.

If production is submodular, however, sorting is negative and assortative in the benchmark and, by revealed preference, workers' wages are higher than the self-match wage.  The sign of $\text{Var}(w_B) - \text{Var}(w_S)$ is then ambiguous and depends (only) on the production function and the distribution of workers' skill. This is easiest to see when skills are binary, in which case wages are not unique in the benchmark and any wage structure of the form
\[w_B(x_1)=0.5(F(\mathbf{x_1})+\alpha_{x_1} (2F(h, l)-F(\mathbf{h})-F(\mathbf{l})),\]
is sustainable in equilibrium. Here, $\alpha_{x_1} \in [0, 1]$ denotes the bargaining power of workers with skill $x_1$, and $\alpha_l+\alpha_h=1$.\footnote{Workers' bargaining power does not matter in this model as long as some workers self-match, because competition reduces the bargaining set to a singleton. If, however, production is submodular and skills are binary, then there is no self-matching in the benchmark and the bargaining set is not a singleton anymore.} Thus, in the binary skills case $\text{Var}(w_B)- \text{Var}(w_S)=0.25(2\alpha_l-1)(2F(h, l)-F(\mathbf{h})-F(\mathbf{l}))$, which is positive if and only if  $\alpha_l \geq 0.5$. In other words, the benchmark wage distribution is less unequal than the self-match one if and only if low-skill workers have a stronger bargaining position in the benchmark.\footnote{The sign of $\text{Var}(w_B) - \text{Var}(w_S)$ remains ambiguous also when benchmark wages are uniquely determined. To see this, consider the case of multiplicative production and log-normally distributed skills (i.e., Assumption \ref{item: mult} with $c<0$ and $\phi(x)=$). In that case, the benchmark and self-match wages become $w_B(x_1)=(Ac+\exp(2c\delta_1)-1)/(2c)+\exp(2c\delta_1)( \ln x_1-2\delta_1)$ and $w_S(x_1)=A+(x_1^{2c}-1)/c$
and thus
$\text{Var}(w_B) - \text{Var}(w_S)=\exp(4\delta_1 c)\left(\omega_{11}- \exp(2 \omega_{11})+\exp(\omega_{11}) \right).$
Clearly, this expression (a) equals to 0 and has a positive derivative for $\omega_{11}=0$ and (b) is concave in $\omega_{11}$. It follows, therefore, that there exists some $\bar{\omega}_{11}$ such that if $\bar{\omega}_{11}<\omega_{11}$ then $\text{Var}(w_B) > \text{Var}(w_S)$ and if $\bar{\omega}_{11}>\omega_{11}$ then $\text{Var}(w_B) < \text{Var}(w_S)$.}

Recall that Proposition \ref{prop: sorting} implies that if low-skill workers have sufficiently stronger relative concerns than high-skill workers then self-matching obtains in equilibrium. Thus, if the benchmark wage distribution is indeed less unequal than the self-match one, then the presence of relative concerns may well increase wage inequality in comparison to the benchmark.\looseness=-1

\begin{cor}\label{cor: var}
Consider an economy $(F, H)$ which satisfies Assumption \ref{ass: suff}.\newline
(i) If $\text{Var}(w_B) < \text{Var}(w_S)$, then  there exists a distribution of traits $H$, such that (a) economy $(F, \tilde{H})$ satisfies Assumption \ref{ass: suff}, (b) the marginal distribution of skill is the same under $H$ and $\tilde{H}$ ($H_{x_1}=\tilde{H}_{x_1}$), and (c) $\text{Var}(w_B)<\text{Var}(w^*)$.   \newline
(ii) Suppose that, in addition, $(F, H)$ satisfies the premise of Proposition \ref{prop: var}. If  $\text{Var}(w_B) \geq \text{Var}(w_S)$, which is always satisfied for supermodular $F$, then  $\text{Var}(w_B) \geq \text{Var}(w^*)$.
\end{cor}

\begin{rem}
 Corollary \ref{cor: var}(i) follows from Proposition \ref{prop: sorting} and the properties of the self-match and benchmark wages, neither of which depend on the fact that all (or nearly all, under Assumption \ref{item: mult}) workers have a positive $x_2$.
 Corollary \ref{cor: var}(ii), in contrast, follows from Proposition \ref{prop: var} and thus requires this additional assumption.
\end{rem}
\begin{rem}
In the binary skill case, inequity aversion equivalent distributions of preferences (see \ref{par: inequity}) are exactly of the form that pushes sorting to be more positive and assortative in skills. In particular, it follows immediately from the discussions in \ref{par: inequity} and \ref{par: sortingbinary} that if $x_3 \geq 0.5-a_F/(1+a_F)$ for all workers, then everyone will self-match in the equilibrium of the model with inequity aversion. This implies that if $a_F<1$ and low-skill workers have higher bargaining power than high-skill workers ($\alpha \geq 0.5$)---and thus the benchmark distribution of wages is less unequal than the self-match one---then \emph{sufficiently strong inequity aversion increases wage inequality in the economy}.

This rather striking result seems to capture a mechanism that may hold much more broadly than just in the labor market: The desire to minimise within-group (here, within-firm) inequality, may push agents to sort with agents who are similar to them. While this indeed eliminates inequity within-groups it \emph{maximises} inequality between-groups---and if the structure of the economy is such that between-group inequality is a greater concern than within-group inequality, it may well increase overall inequality.

\end{rem}

\subsubsection{Welfare}\label{sec: util}
In order to meaningfully compare workers' welfare in this model with the benchmark case, we need to keep workers' preferences constant. One way of accomplishing this is to assume that the $x_2$ considered in the analysis so far is but the product of the `true' preference (denoted by $\tilde{x_2})$ and the intensity of interaction between co-workers $p \in [0, 1]$, with $x_2=p\tilde{x_2}$.
To fix ideas, suppose that $p=0$ corresponds to a case where the production process is fully remote and anonymous, so that the co-workers do not even know who they work with, and $p=1$ corresponds to a fully in-office, team-work based production. This setup allows us to fix the distribution of skill and `true' preference $(x_1, \tilde{x}_2)$ and ask whether an increase in the intensity of interactions from $p=0$ to $p>0$ is welfare-improving. \looseness=-1

\begin{prop}\label{prop: util}
    Suppose that Assumptions \ref{ass: common} and \ref{ass: copula} are satisfied.\newline
    (i) If $F$ is supermodular, then $\text{Pr}(u^*(X_1, X_2) \geq u_B(X_1))=1$. If, in addition, Assumption \ref{item: additive} or \ref{item: mult} is satisfied and $H$ has full support, then $\text{Pr}(u^*(X_1, X_2) > u_B(X_1))=1$.\newline
    (ii) If $F$ is strictly submodular and there exists some $\mathbf{x}^{\prime}$ such that $\mu_1(\mathbf{x}^{\prime})=x_1^{\prime}$, then $\text{Pr}(u^*(X_1, X_2)>u_B(X_1))<1$.
\end{prop}
\deferred[proof: proputil]{\paragraph*{Proof of Proposition \ref{prop: util}}
The first part of (i) is immediate, because a worker of an arbitrary type $\mathbf{x}$ can guarantee themselves the benchmark payoff of $0.5F(x_1, x_1)=u_B(x_1)$ by self-matching. The second part follows, because if either Assumption \ref{item: additive} or \ref{item: mult} is satisfied and $H$ has full support, then---by \eqref{eq: sortingadd} and \eqref{eq: sortingmult}---only a measure zero of workers self-matches and conditional on $x_1$ utility is minimised for the self-matching workers by Proposition \ref{prop: utilandwage} and footnote \ref{foot: generic}.\looseness=-1

If surplus is strictly submodular, then the above logic breaks down, because $u^*(x_1, x_2^*(x_1))=0.5F(x_1, x_1)<u_B(x_1)$; that is, while self-matching is still any workers' outside option, this outside option is strictly worse than their utility under the benchmark. Clearly then, $u^*(\mathbf{x}^{\prime})<u_B(x_1')$, and the result follows from the absolute continuity of $\text{Pr}(x_2|X_1=x_1)$ and continuity of $u^*(\bullet)$. \looseness=-1
}
The welfare impact of social comparisons hinges on the properties of the production function. Under supermodular production, self matching always  guarantees the benchmark payoff, regardless of how intense social interactions are---by revealed preference workers must therefore benefit from social comparisons. Under submodular production, the negative and assortative sorting pattern that prevails in the benchmark allows workers to earn higher payoffs than the self-match payoff: If, therefore, under $p=1$ some workers decide to self-match in order to avoid negative social comparisons, then they must be worse off than they would be under $p=0$. This indicates that if each team could choose $p$, then all teams would choose $p=1$ under supermodular production, but some would chose $p=0$ under submodular production. This observation forms the basis of the `theory of the firm' outlined below. \looseness=-1

\section{A Unified Theory of Changes in Sorting, Outsourcing and Wage Inequality}\label{sec: ToF}

In this Section, I show that the interaction between skill-biased technological change and relative concerns helps to explain a number of seemingly disconnected but well-documented empirical trends. Namely, and starting with the least obvious (a)  the marked increase in domestic outsourcing \citep{Goldschmidt2017, Bergeaud2024}, (b) the increase in workers' sorting, that is,  the probability that high-skill workers' co-workers are high-skill \citep{Freund2022}, (c) the disproportionately large increases in between- compared to within-firm inequality \citep{Song2018, Tomaskovic2020} and (d) the strong increase in overall wage inequality \citep{Bound1992, Katz1992, Juhn1993}.

Let me start by defining \emph{skill-biased technological change} (SBTC).

\begin{defi}
A technological change is a change in the production function from some $F(\cdot, \cdot; \theta_1)$ to some $F(\cdot, \cdot; \theta_2)$. A technological change is \emph{skill-biased} if there exists an increasing function $S: I_{x_1} \to \mathbb{R}$ such that $F(x_1^k, x_1^j; \theta_2)=F(x_1^k, x_1^j; \theta_1)+S(x_1^k)+ S(x_1^j)$.
\end{defi}
This definition of SBTC, while well-grounded in the literature \citep{Lindenlaub2017}, is not entirely general, in that it restricts attention to technological change that is own-skill biased. In other words, the additional output produced by worker $\mathbf{x}$ does not at all depend the co-worker's skill, only on the workers individual ability. The advantage of this definition, however, is that in the benchmark SBTC as defined here leaves sorting unaffected and does not trigger any trickle-down effect; this will help me highlight the impact that the interaction between SBTC and relative concerns has on sorting and inequality.

\subsection{A Theory of the Firm}
While the model as presented already contains the ingredients needed to study the impact of SBTC on sorting, as well as on overall- and between-firm inequality, in order to address its impact on outsourcing I need to allow the teams to choose where to draw the boundary of the firm. The basic premise is very simple: Wage comparisons weigh lighter in agents' utility when they happen across firm boundaries. In other words, the co-worker's high wage bothers the worker less if the worker is a subcontractor rather than a subordinate.

Suppose that each matched team has the option of \emph{outsourcing}, that is, forming two separate firms instead of one.
As is standard in the theory of the firm literature, outsourcing comes at a cost $c \geq 0$: Contracts need to be written, there is additional accounting, etc. The possible advantage of outsourcing, however, is that the two co-workers stop comparing their wages, so that each workers utility is equal to their own wage only, with $u(w^k,w^j)=w^k$. If the team decides not to outsource, then each co-worker's utility is as in the baseline model.
As this extended model is significantly less tractable than the baseline (e.g., it has no TU-representation), I restrict attention to the binary-skill case---which remains simple enough to solve---throughout. \looseness=-1

\begin{prop}\label{prop: outsourcing}
Suppose that Assumption \ref{item: binary} is satisfied, and denote by $s_F$ the loss of surplus resulting from self-matching, with
$s_F\equiv F(h, l)-0.5(F(\mathbf{h})+F(\mathbf{l}))$.\newline
(i) If $c>s_F$, then there is no outsourcing and the equilibrium is as described in \ref{par: sortingbinary}.\newline
(ii) If $c \in [0, s_F]$, then all teams formed in any equilibrium are between a high- and a low-skill worker. Worker $(\mathbf{x})$ who forms a (non-)outsourcing team receives  $u^o(\mathbf{x})$ ($u^n(\mathbf{x})$) and earns a wage $w^o(\mathbf{x})$ ($w^n(\mathbf{x})$), where:
\begin{IEEEeqnarray}{rCl}
    u^o(\mathbf{x})&=&w^o(\mathbf{x})=\frac{F(\mathbf{x_1})}{2}+\alpha_{x_1}(s_F-c), \,
    w^n(\mathbf{x})=\frac{F(h, l)}{2}+\frac{2w^o(\mathbf{x})-F(h, l)}{4T_{x_1}(y^o)}, \label{eq: outsourcingwages} \\
    u^n(\mathbf{x})&=&\frac{F(h, l)}{2}+(2w^n(\mathbf{x})-F(h, l))(1+2x_2).\label{eq: payoffs}
\end{IEEEeqnarray}
Here, $\alpha_{x_1} \in [0, 1]$ denotes the bargaining power of workers with skill $x_1$ (with $\alpha_l+\alpha_h=1$) and  $y^o$ uniquely solves
\begin{equation}\label{eq: ycirc}
\min\{\frac{T_h(1)}{T_l(1)}, \max\{\frac{T_h(0)}{T_l(0)}, 1-\frac{2c}{F(h,l)-F(\mathbf{l})+2\alpha_l(c-s_F)}\}\}=\frac{T_h(y^o)}{T_l(y^o)}.
\end{equation}
Low-skill workers are outsourced  iff their $x_2 > T_l(y^o)-0.5$. Outsourced low-skill workers match with high-skill workers of $x_2<T_h(y^o)-0.5$, the non-outsourced low-skill workers match with high-skill workers of $x_2\geq T_h(y^o)-0.5$.

\end{prop}

\deferred[proof: propoutsourcing]{\paragraph*{Proof of Proposition \ref{prop: outsourcing}}
(i) In this case, outsourcing is dominated by self-matching: In choosing between these options relative concerns do not matter, and  $c> s_F$ implies that the output from self-matching is higher than from outsourcing. %

(ii) By the same logic as in (i), $c<s_F$ implies that outsourcing dominates self-matching for all teams, and thus all teams consist of one high- and one low-skill worker. In equilibrium, $w^n(\mathbf{x}), w^o(\mathbf{x})$ are constant in $x_2$; otherwise, no-one would match with the workers earning $\max_{x_2} w^i(x_1, x_2)$ for $i \in \{n, o\}$ and $x_1 \in \{h, l\}$. Any high-skill worker can guarantee themselves the self-match wage of $F(\mathbf{h})/2$, hence $\Delta w^n\equiv w^n(h, x_2)-w^n(l, x_2) = 2w^n(h, x_2)-F(h, l) \geq F(\mathbf{h})-F(h, l)>0$. Accordingly, $u^n(h, x_2)$ ($u^n(l, x_2)$) is increasing (decreasing) in $x_2$, whereas $u^o(\mathbf{x})$ is constant in  $x_2$. By feasibility, the measure of (non-)outsourcing high-skill workers must be equal to the measure of (non-)outsourcing low-skill workers; thus, there exists some $y^o$ such that a high- (low-)skill worker outsources (gets outsourced) iff $x_2 < T_h(y^o)-0.5$ ($x_2 > T_l(y^o)-0.5$). Of course, the workers $(x_1, T_{x_1}(y^o)-0.5)$ are indifferent between outsourcing and not, so that
\begin{equation}\label{eq: nonoutwage}
w_{x_1}^n=\frac{w_{x_1}^o+(T_{x_1}(y^o)-0.5)F(h, l)}{2T_{x_1}(y^o)}.
\end{equation}
Finally, by the same logic as in Section \ref{sec: inequality}, the wages and payoffs under outsourcing depend on the bargaining power of each type of agents,
which---together with \eqref{eq: nonoutwage}---yields \eqref{eq: outsourcingwages}. \eqref{eq: payoffs} follows from substituting \eqref{eq: outsourcingwages} into \eqref{eq: utildefinition}. \eqref{eq: ycirc} follows then from adding up \eqref{eq: nonoutwage} for the two skill types, using $w_h^n+w_l^n=F(h,l)$ and $w_h^o+w_l^o=F(h,l)-c$, and substituting for $w_{x_1}^o$ from \eqref{eq: outsourcingwages}.

}
In the baseline model, welfare-reducing social comparisons can be avoided only through self-matching. Outsourcing provides an alternative way of opting-out from these comparisons. For that alternative to be used, it must be cheaper than self-matching. As self-matching is output-maximising under supermodular production, outsourcing can happen only if production is submodular ($s_F>0$), and the cost of outsourcing is low compared to the loss of production stemming from self-matching ($c< s_F$).

To understand which teams outsource, consider the \emph{marginal team}, that is a team which is indifferent between forming one or two firms. If outsourcing comes at cost, their indifference implies that social comparisons are costly within that team, and thus the high-skill worker has weaker relative concerns than the low-skill worker.
In the extreme case of cost-less outsourcing, the high- and low-skill workers forming the marginal team have equally strong relative concerns. Of course, the high-skill workers forming (non-)outsourcing teams have weaker (stronger) relative concerns than the high-skill worker in the marginal team; and \emph{vice versa} for the low-skill workers.

 Wages and payoffs are only determined up to a constant: Any split of the additional output $s_F-c$ resulting from outsourcing can be supported in equilibrium. Outside of that, the wages and payoffs have similar features as in the baseline model.  High-skill workers earn lower wages in non-outsourcing than in an outsourcing team, whereas low-skill workers earn higher wages. Nevertheless, the utility of the workers employed by a non-outsourcing team is higher than the utility of workers with the same skill employed by an outsourcing team.\looseness=-1

\subsection{Consequences of Skill-Biased Technological Change}

In this section, I explain how the interaction between relative concerns and SBTC can explain the recent evolution of outsourcing, sorting and  inequality. I focus on the case in which outsourcing and non-outsourcing teams co-exist; that is, I assume $c<s_F$ and $T_l, T_h$ such that $y^o \in (0, 1)$. \looseness=-1

\subsubsection{Outsourcing and Sorting}\label{sec: outsourcing}

As SBTC leaves $s_F$ unchanged, it follows directly from  \eqref{eq: ycirc} in Proposition \ref{prop: outsourcing} that SBTC increases $y^o$ and thus also the number of outsourcing teams. In other words, as long as any jobs were outsourced initially, skill-biased technological change will cause more outsourcing.
To understand the intuition, recall that in the marginal team the high-skill has stronger relative concerns than the low-skill worker. SBTC further increases the inequality within that team---and with that the welfare loss from social comparisons. As a result, the marginal team now strictly prefers to outsource, and the number of outsourcing teams increases. Furthermore, the wages of the newly outsourced low-skill workers fall in comparison to non-outsourced low-skill workers, which is consistent with the empirical findings from \cite{Goldschmidt2017} and \cite{Bergeaud2024}.

If $c<s_F$ then all production teams consist of one high- and one low-skill worker, and thus SBTC has no effect on how workers sort into \emph{production teams}. Crucially, however, SBTC does affect how workers sort into \emph{firms}, because every outsourcing team consists of two single-worker firms!
 Thus, for an econometrician who observes the composition of firms but not teams, workers from outsourcing teams are sorted positively and assortatively, whereas workers in non-outsourcing firms are negatively sorted. It follows, therefore, that increase in outsourcing caused by SBTC results in workers sorting more positively into firms.\looseness=-1

\subsubsection{Wage Inequality}\label{sec: outsourcinginequality}
 Denote within- (between-) firm inequality by WFWI (BFWI), and the difference between high- and low-skill wages in (non)outsourcing teams by $\Delta w^o$ ($\Delta w^n$). Within-firm inequality features only in non-outsourcing firms, which form with probability $1-y^o$, so that\looseness=-1
\begin{equation} \label{eq: wfwi}
    \text{WFWI}=0.25(1-y^o)(\Delta w^n)^2.
\end{equation}

Between-firm wage inequality itself consists of two components: the difference in the output produced by outsourcing and non-outsourcing teams, and the difference in wages of high- and low-skill workers in outsourcing teams, with
\begin{equation}\label{eq: bfwi}
   \text{BFWI}=0.25\left[y^o(1-y^o)c^2+y^o\left(\Delta w^o \right)^2.\right]
\end{equation}
It follows thus from s \eqref{eq: wfwi}-\eqref{eq: bfwi}, the law of total variance and some algebra that
\begin{IEEEeqnarray}{rCl}
 \text{Var}(W)&=&0.25\left[(1-y^o)(\Delta w^n)^2 +y^o(1-y^o)c^2+y^o\left(\Delta w^o\right)^2\right] \label{eq: var}\\
 \text{BFWI}/\text{WFWI}&=& y^o\left(T_l(y^o)-T_h(y^o)\right)^2+(\frac{y^o}{1-y^o})\left(T_h(y^o)+T_l(y^o) \right)^2. \label{eq: ratio}
\end{IEEEeqnarray}

\begin{prop}\label{prop: inequality}
Suppose that Assumption \ref{item: binary} is satisfied and that $c< s_F$. \newline
(i) If $\underline{x}_2 \geq 0$, then skill-biased technological change increases the variance of wages.\newline
(ii) If $\deriv{y}\ln (G_l^{-1}(y)) \leq 4\sqrt{3}/9$ for all $y \in [0, 1]$, then SBTC increases $\text{BFWI}/\text{Var}(W)$. \newline
\end{prop}
\deferred[proof: propinequality]{\paragraph*{Proof of Proposition \ref{prop: inequality}}
In order to take derivatives, define $F(\cdot; \theta)\equiv \theta F(\cdot; \theta_2)+(1-\theta)F(\cdot; \theta_2)$. We can then write
\[
   4 \tderiv{\theta} \text{Var}(W)= y^o\deriv{\theta} \Delta (w^n)^2 +(1-y^o)\deriv{\theta}\left(\Delta w^o \right)^2
   - \underbrace{\deriv{\theta}y^o}_{> 0} \underbrace{\left[(\Delta w^n)^2+(1-2 y^o)c^2-1 \right]}_{\text{Sorting Effect}}.
\]

As $\Delta w^o=0.5(F(\mathbf{h})-F(\mathbf{l})-(1-2 \alpha)(s_F+2c)$ it follows from the definition of SBTC that $\deriv{\theta}\left(\Delta w^o \right)^2>0$. Similarly, it is easy to show that we can write $\Delta w^n=0.5(F(h,l)-F(\mathbf{l})+\alpha(c+s_F/2))/T_l(y^o)$ from which we can also immediately see that $\deriv{\theta}\left(\Delta w^n \right)^2>0$. What is left is to show that the sorting effect is positive as well. That this is the case as long as $\underline{x_2}\geq 0$ can be easily verified by noticing that  $\Delta w^n=(T_l(y^o)+T_h(y^o))^{-1}\Delta w^o$ and $c=\Delta w^n\left(T_l(y^o)-T_h(y^o) \right)$ and using elementary algebra.\looseness=-1

(ii)
Denote $(T_l(y^o)-T_h(y^o))/(T_h(y^o)+T_l(y^o))$ by $B$ (note that $B \leq 1$) and define $P\equiv(\deriv{y}BFWI/\text{Var}(W))(T_h(y^o)+T_l(y^o))^2$. Then:
\begin{IEEEeqnarray}{rCl}
  P& \geq &  \frac{2By^o(\deriv{y}T_l(y^o)-\deriv{y}T_h(y^o))}{T_h(y^o)+T_l(y^o)}-\frac{1}{(1-y^o)^2}+\frac{2(\frac{y^o}{1-y^o})(\deriv{y}T_l(y^o)+\deriv{y}T_h(y^o))}{T_h(y^o)+T_l(y^o)} \nonumber\\
  &\geq & 4((2-y^o)y^o/y)\left(\deriv{y}T_l(y^o)/(0.5+\underline{x_2}+T_l(y^o))+1/(2(1-y^o)(1-(1-y^o)^2))\right ).\label{eq: claim1}
\end{IEEEeqnarray}
 (ii) follows because $1/2(1-y^o)(1-(1-y^o)^2))>4 \sqrt{3}/9$ and $\deriv{y}\ln\left(0.5+\underline{x_2}+T_l(y^o)\right) \leq -\deriv{y} \ln( G_l^{-1}(y^o))$.

}

 By inspection of \eqref{eq: var}, SBTC has both a direct effect (through changes in $\Delta w^n$ and $\Delta w^o$) and an indirect, sorting effect (through changes in $y^o$) on wage inequality. The direct effect must increase wage inequality. The sign of the sorting effect is ambiguous in general, but must be positive as long as outsourcing is viable and all workers have positive relative concerns; thus, SBTC increases wage inequality.\footnote{Interestingly, if production was supermodular, and thus outsourcing was not viable, both the sorting and the overall effects could be negative.}

By  \eqref{eq: ratio}, SBTC affects the ratio of between- to within-firm wage inequality only through sorting $y^o$. In general, the link between sorting and said ratio is  ambiguous. On the one hand, an increase in sorting corresponds to more outsourcing, and outsourcing mechanically transforms within- into between-firm inequality. On the other hand, stronger sorting means that the marginal team hires a low-skill worker of weaker relative concerns than before, which increases between-firm inequality and decreases within-firm inequality. Proposition \ref{prop: inequality}(ii) states that if low-skill workers' relative concerns are not too differentiated, the positive effect is much stronger and must dominate.
This mirrors the results from \cite{Gola2023}: In the absence of compensating differentials, SBTC increases the ratio of between-firm and overall inequality; when compensating differentials are present, then SBTC's impact is ambiguous.\looseness=-1

\subsection{Discussion}

\paragraph*{Social Comparisons Across Firm Boundaries} I assumed that social comparisons within a production team are much weaker if that team is split into two firms. \cite{Nickerson2008} attribute this weakening of social comparisons across firm boundaries to the salience of within-firm comparisons, and to within-firm competition for the same resources. They also provide a number of persuasive case studies in which the firm boundary mattered critically for the strength of social comparisons. Another justification of this assumption can be derived from \cite{Coase1937}, who hypothesised that some people like to direct others, and some like to be directed, and differentiated between `employees' and `subcontractors' precisely by the degree to which their work is directed. Under this interpretation, co-workers in outsourcing teams work together, but---in contrast to a non-outsourcing firm---none of them is directed by the other, and thus social comparisons matter less.

\paragraph*{Theory of the Plant vs. Firm} A compelling feature of this extended model is that it provides both a `theory of plant' (sorting into production teams) and  `theory of firm' (a team's decision whether to form one or two firms), and that production equivalent `plants' draw their firm boundaries differently. Furthermore, the `plant'- and firm-formation decisions interact in this model. As outsourcing becomes viable, a high-skill worker who would have previously self-matched, switches to having an outsourced low-skill co-worker. This implies that having the option of cheaply redrawing the boundary of the firm affects what ``plants'' are formed. \looseness=-1

\paragraph*{Submodular Production} The condition $c \in (0, s_F)$ is satisfied only when the production function is submodular. This is potentially problematic,  because supermodular production functions are more commonly assumed in the sorting literature. However, this is largely because the empirical correlation in the level of co-workers skill is large and positive \citep[see Figure 1(b) in][for example]{Freund2022}, a fact which in standard models can be reproduced only with supermodular production. In my model, however, submodular production is perfectly consistent with positive and assortative matching in skills, as long as low-skill workers have stronger relative concerns than high-skill workers (see Section \ref{sec: sorting}). Furthermore, (locally) submodular production has been convincingly microfounded by \cite{Kremer1996} as the by-product of workers' self-selection into roles within the firm, and more recently by \cite{Boerma2021} as the outcome of within-team problem solving. \looseness=-1

\paragraph*{Lower Cost of Outsourcing} An obvious alternative explanation for the trends in outsourcing, sorting and inequality is a decrease in the cost of outsourcing. One can see immediately that this would have the same qualitative impact on outsourcing, sorting and the ratio of between- to within-firm inequality as SBTC in my model. Its impact on wage inequality is ambiguous, however; while the change in sorting it triggered by lower $c$ increases inequality, the lower cost of outsourcing also directly decreases the  inequality between the group of teams that outsource and those that do not.

\paragraph*{Inequity Aversion} Inequity aversion equivalent preferences satisfy the premises of Propositions \ref{prop: outsourcing} and \ref{prop: inequality}(ii)-(iii). Thus, if the cost of outsourcing is neither too high nor too low (so that both outsourcing and non-outsourcing teams co-exist), then the impact of SBTC on outsourcing, sorting and the ratio of between-firm to overall inequality under inequity aversion is much the same as under relative concerns. Under inequity aversion, however, both the sorting and overall effect of SBTC on wage inequality can be negative. This is because, perversely, the difference in high- and low-skill wages may actually be greater in non-outsourcing than in outsourcing firms:  If the high-skill worker in the marginal team is more averse to inequity than their co-worker is, then the high-skill worker has to receive \emph{higher compensation for the inequity they must endure} than their low-skill co-worker. If this is the case, then the increase in outsourcing decreases the number of teams with elevated wage inequality, thus contributing to a decrease in wage inequality.

\paragraph*{Remote Work} Instead of outsourcing, teams could escape detrimental social comparisons by re-organising production. For example, teams could choose to work remotely, which would decrease the intensity of social interactions. Most results from this section can be reinterpreted as results about the prevalence of remote work. The one exception are the results about sorting. When firms avoid social comparisons by re-organising production rather than by outsourcing, then measured and real sorting coincide, and thus SBTC has no impact on measured sorting.

\section{Concluding Remarks}\label{sec: cr}
In this paper, I develop a one-sided assignment model in which workers differ in skill and the strength of their relative concerns. The heterogeneity of relative concerns makes the problem naturally two-dimensional and implies that utility is imperfectly transferable. Yet, I am able to fully characterise the equilibrium for a large class of cases, leveraging the facts that the problem admits a transferable utility representation and that the distribution of traits of `workers' must be the same as that of `co-workers' in equilibrium.

The existence of a transferable utility representation implies that equilibrium sorting optimally trades off output maximisation with the need to maximise the welfare gain stemming from within-team social comparisons. The latter is accomplished by matching high-skill workers to co-workers with weak relative concerns. As a result, equilibrium sorting can be positive (negative) assortative in skill even when production is submodular (supermodular). When production is supermodular, then---in comparison to a model without relative concerns---all workers are better off and wage inequality is lower. With submodular production, however, the presence of heterogenous relative concerns may increase wage inequality and make high-skill workers with low relative concerns and low-skill worker with strong relative concerns worse off.\looseness=-1

Finally, I build on that last insight to argue that skill-biased technological change may have caused the observed long-term increase in domestic outsourcing. Following \cite{Nickerson2008}, I assume that the salience of social comparisons weakens if one of the team-members is outsourced. If that is the case, then teams consisting of high-skill workers with low relative concerns and low-skill worker with strong relative concerns would like to outsource the low-skill worker, even though outsourcing is costly. Skill-biased technological change increases within-team inequality and thus increases the cost of keeping the low-skill worker in-house for such teams; as a result, the number of outsourcing teams increases.

\newpage
\begin{appendix}

\section{Omitted Proofs and Derivations}\label{app: proofs}
\shownow{proof: assumptions}
\shownow{deriv: truthtelling}
\shownow{proof: theo1}
\shownow{proof: propsorting}
\shownow{proof: propvar}
\shownow{proof: proputil}
\shownow{proof: propoutsourcing}
\shownow{proof: propinequality}
\end{appendix}

\bibliographystyle{chicago}
\bibliography{unified}

\end{document}